\documentclass[journal,12pt,onecolumn,draftclsnofoot]{IEEEtran}
\usepackage{ifpdf}
\usepackage{cite}
\ifCLASSINFOpdf
  \usepackage[pdftex]{graphicx}
\else
\fi
\usepackage{amsmath}
\usepackage{amsfonts}
\usepackage{amssymb}
\usepackage{amsthm}
\usepackage{color}
\usepackage[usenames,dvipsnames,svgnames,table]{xcolor}
\usepackage{algorithmic}
\usepackage{algorithm}
\usepackage{array}
\usepackage{url}
\usepackage{float}
\usepackage{graphicx}
\usepackage{subfigure}

\newtheorem{theorem}{Theorem}

\newtheorem{example}{Example}

\newtheorem{corollary}{Corollary}
\newtheorem{lemma}{Lemma}
\newtheorem{remark}{Remark}

\newcommand{\cgraf}{{\mathcal G}}
\newcommand{\cV}{\mathcal{V}}
\newcommand{\cB}{\mathcal{B}}
\newcommand{\cS}{\mathcal{S}}

\DeclareMathOperator{\lcm}{lcm}

\newcommand{\cP}{\mathcal{P}}
\newcommand{\bm}[1]{\mathbf{#1}}

\newcommand{\mc}[1]{{\mathcal #1}}
\newcommand{\prob}[1]{\mbox{Prob}\left[#1\right]}
\newcommand{\vol}[1]{\mbox{vol}\left[#1\right]}
\begin{document}
\title{On Communication for Distributed Babai Point Computation}%  approximate nearest lattice point in a distributed network}
%\title{Communication Efficient Distributed Babai Point Computation}
\author{Maiara~F.~Bollauf,
        Vinay~A.~Vaishampayan,
        and~Sueli~I.~R.~Costa.
\thanks{M. F. Bollauf and S. I. R. Costa are with the Institute of Mathematics, Statistics and Computer Science, University of Campinas (e-mail: bollauf@ieee.org, sueli@unicamp.br).}% <-this % stops a space
\thanks{V. A. Vaishampayan is with Department of Engineering Science and Physics, City University of New York (CUNY) (e-mail: Vinay.Vaishampayan@csi.cuny.edu).}
%\thanks{S. I. R. Costa is with the Institute of Mathematics, Statistics and Scientific Computing, University of Campinas (e-mail: sueli@iunicamp.br).}
\thanks{This paper was presented in part at the IEEE International Symposium on
Information Theory, Aachen, Germany, 2017 \cite{BVC:2017}.}}

\maketitle


\begin{abstract}
We present a communication-efficient distributed protocol for computing the Babai point, an approximate nearest point for a random vector ${\bf X}\in\mathbb{R}^n$ in a given lattice. 
We show that the protocol is optimal in the sense that it minimizes the sum rate when the components of $\bm{X}$ are mutually independent. We then investigate the error probability, i.e. the probability that the Babai point does not coincide with the nearest lattice point. In dimensions two and three, this probability is seen to grow with the packing density.  For higher dimensions, we use a bound from probability theory to estimate the error probability for some well-known lattices.  Our investigations suggest that for uniform distributions, the error probability becomes large with the dimension of the lattice, for lattices with good packing densities. We also consider the case where  $\bm{X}$ is obtained by adding Gaussian noise to a randomly chosen lattice point. In this case, the error probability goes to zero with the lattice dimension when the noise variance is sufficiently small. In such cases, a distributed algorithm for finding the approximate nearest lattice point is sufficient for finding the nearest lattice point.  
\end{abstract}


\IEEEpeerreviewmaketitle

% \fi

\section{Introduction}

We are given a lattice $\Lambda \subset \mathbb{R}^n$ and a random vector  of observations, ${\bf X}=(X_1,X_2,\ldots,X_n)\in \mathbb{R}^n$. Each $X_i$ is available at a distinct sensor-processor node (SN), which is connected by a communication link to a central computing node (CN). The objective is to compute at the CN, the Babai point, a well-known approximation to the nearest lattice point of ${\bf X}$~\cite{babai}. Towards this end, the $i$th SN sends an approximation of $X_i$ to the CN at a communication rate of $R_i$ bits/sample. In this work, we present a communication protocol for this computation and show that it is optimal in the sense of minimizing the communication rate. We then investigate the connection between the structure of the lattice, as determined by its generator matrix, and the communication cost, the error probability (the probability that the Babai point does not coincide with the nearest lattice point) and the packing density. While this connection is of independent interest, it also allows a designer to understand situations under which any further communication for determining the true nearest lattice point is unnecessary.  %, since under some assumptions on the probability distribution of ${\bf X}$, there is a high probability that the Babai point coincides with the nearest lattice point. 
Our model for distributed computation, is referred to as the  centralized model, and  is illustrated in Fig.~\ref{fig-p1}.

We note that our problem is a special case of the general distributed function computation problem, where the objective is to compute a given function $f(X_1,X_2,\ldots,X_n)$ at the CN based on information communicated from each of the $n$ SN's~\cite{Orlitsky:2001}. 
In our case, $f$ is the function which computes an approximate nearest lattice point based on the  nearest plane algorithm \cite{babai} and $f({\bf X})$ is  the Babai point.

\begin{figure}[H]
\begin{center}
\includegraphics[scale=0.28]{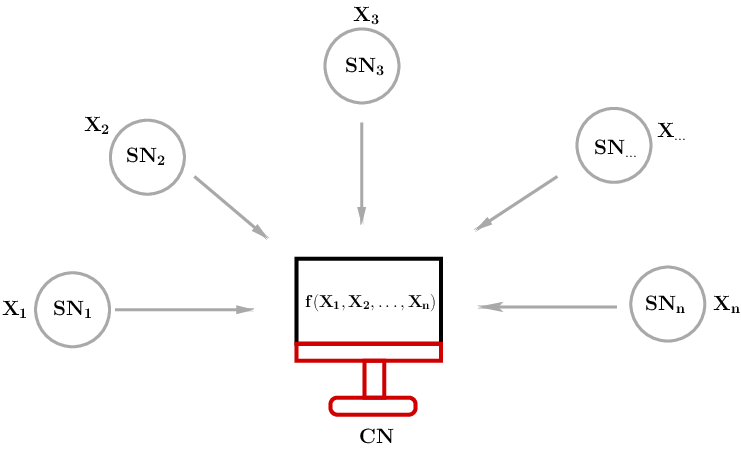}
\caption{Centralized model for distributed computation. Each sensor node (SN) encodes its observation at a finite rate and sends it to the central compute node (CN), where the function $f$ is to be computed. The problem is to determine the tradeoff between communication rate and the accuracy with which the function is computed. In this work, the function is the approximate nearest lattice point (Babai point).}\label{fig-p1}
\end{center}
\end{figure}

Interest in communication issues for the distributed computation of the Babai point, and more generally for the nearest lattice point~\cite{VB:2017}, arise in many contexts: wireless communication, machine learning and cryptography. We briefly describe the applications next.

In MIMO wireless systems, the decoding problem is equivalent to finding a nearest lattice point. Well-known systems such as V-BLAST prefer to find the Babai point because of the high computational complexity of finding the nearest lattice point. Thus, distributed computation of the Babai point is useful in distributed MIMO receivers~\cite{VVMIMO:2020}. More generally, communication issues for channel decoding and demodulation have been studied in the context of cooperative communications~\cite{Draper:1998},~\cite{Wang:2007}.  For a comprehensive review of lattice methods in communication, see~\cite{zamir2014}.

In recent years, interest has grown in communication issues related to distributed machine learning~\cite{LASY:2014}. Such problems also fit into the distributed function computation framework, and we expect that lattice methods will eventually play an important role here.%In the area of  distributed source coding in a sensor network, a potential application is  for computing a nearest lattice vector in a lattice codebook~\cite{zamir2014}, when the observations are made at remote sensor nodes. In a similar spirit
%The process of finding the closest lattice point is widely used for decoding lattice codes~\cite{agrelletal}. Lattice coding offers significant coding gains for noisy channel communication~\cite{conwaysloane}  and for quantization~\cite{berger}. 
	
%	 computing and wireless communications, it may be necessary for a vector of measurements to be available at distinct locations. In order to reduce network bandwidth usage, it is logical to consider a vector quantized (VQ) representation of these measurements, subject to a fidelity criterion and for once a VQ representation is obtained, it can be forwarded in a bandwidth efficient manner to other parts of the network. However, there is a communication cost to obtain the vector quantized representation and this paper is our attempt to understand the tradeoffs involved.

The study of the approximate nearest lattice point is also of interest in cryptography. In fact the  nearest lattice point problem has been proposed as a basis for lattice cryptography \cite{Ajtai:1996},\cite{galbraith},\cite{mathcrypto},\cite{MicGold:2012},\cite{peikert}, due to its hardness \cite{emdeboas}, examples being the GGH and LWE cryptosystems. The security of such cryptossystems rely on the solution of this problem and the nearest plane algorithm is used to estimate the resistance to attack when the received message is relatively close to the lattice point to be decoded.  Our work is of interest in understanding the communication required in a distributed lattice-based cryptosystem. 

This paper is based on preliminary work presented in~\cite{BVC:2017}.

%	Since random variables are real valued, these calculations would require that the system communicate an infinite number of bits to compute $f$ exactly, but the network has limited capacity, thus the information must be quantized in a suitable manner. This quantization, however, affects the accuracy of the function that we are trying to calculate. Therefore, the main goal of this paper is to manage the tradeoff between communication cost and function computation accuracy. 	
%	
%	
%	Applications to this problem arise in MIMO systems~\cite{RCP:2009}, and network management in wide area networks~\cite{KCR:2006}, to name a few. Information theory~\cite{Shannon:1948} has resulted in tight bounds,~\cite{Orlitsky:2001},~\cite{MaIshwar}. We particularly emphasize the lower bounds for the success probability of the Babai point to coincide with the closest lattice point, derived when the received vector $\bm{x} \in \mathbb{R}^n$ follows a Gaussian distribution \cite{Chang:2013, Teu:1998, Wen:2017}. Here, we also discuss the success probability (or analogously error probability), however our assumption is that the received vector follows an uniform distribution instead.

The paper is organized as follows. Mathematical preliminaries are in Sec.~\ref{sec1}. A communication protocol and its associated communication cost are presented in Sec.~\ref{sec3}, along with a proof of optimality.  The error probability,% i.e. the probability that the Babai point does  not coincide with the true nearest lattice point, 
~is analyzed in Sec.~\ref{sec4} for dimensions two and three, for a uniform conditional distribution on ${\bf X}$. This requires a special basis for a lattice as described  in Sec~\ref{secMOS}. This section also examines the relation between the error probability and the packing density of the lattice being considered.  Since these calculations are difficult to generalize to higher dimensions, we use probabilistic tools to understand the behavior of the error probability and its relation to the `sphericity' of a Voronoi cell of the lattice in Sec.~\ref{sec5}, in terms of its covering and packing radii. In this section we also discuss and compare results about error probability and packing density when ${\bf X}$ is obtained by adding  Gaussian noise to a randomly chosen lattice point.  Conclusions and future work are in Sec.~\ref{secC}.

\section{Lattice Basics and Preliminary Calculations} 
\label{sec1} 

Notations, lattice basics and error probability simplifications are described in this section. 

	A (full rank) lattice $\Lambda \subset \mathbb{R}^n$ is the set of all integer linear combinations of a set of linearly independent vectors $\{\bm{v_1},\bm{v_2},\ldots,\bm{v_n}\} \subset \mathbb{R}^n,$ called \textit{lattice basis}. We can also write $\Lambda=\{V\bm{u},~\bm{u} \in \mathbb{Z}^n\},$ where the columns of the \textit{generator matrix} $V$ are the basis vectors $\bm{v_{1}}, \dots, \bm{v_{n}}$. The matrix $A=V^{T}V$ is the associated \textit{Gram matrix} and the $(i,j)$ entry of $A$ is the Euclidean inner product of $\bm{v_{i}}$ and $\bm{v_{j}},$ which here will be denoted by $\bm{v_{i}} \cdot \bm{v_{j}}.$  

	A set $\mathcal{F}$ is called a \textit{fundamental region} of a lattice $\Lambda$ if all its translations by elements of $\Lambda$ cover $\mathbb{R}^{n},$ i.e., $\underset{\lambda \in \Lambda}{\bigcup} \mathcal{F} + \lambda = \mathbb{R}^n$ and the interior of $\lambda_1 + \mathcal{F}$ and $\lambda_2 + \mathcal{F}$ do not intersect for $\lambda_1 \neq \lambda_2.$  The \textit{Voronoi region} or \textit{Voronoi cell} $\cV(\lambda)$ is an example of fundamental region and it is defined as $$\cV(\lambda)=\{\bm{x} \in \mathbb{R}^{n}: ||\bm{x}-\bm{\lambda}|| \leq ||\bm{x}-\bm{\tilde{\lambda}}||, \ \text{for all} \ \bm{\tilde{\lambda}} \in \Lambda\},$$ where $||.||$ denotes the Euclidean norm. Note that  $\cV(\lambda)$ is congruent to $\cV(0)$. The \textit{volume} of a lattice $\Lambda$ is the volume of any of its fundamental regions and it is given by $\text{vol}(\Lambda)=|\det(V)|,$ where $V$ is a generator matrix of $\Lambda.$ We refer to $\cV(\lambda)$ as a \emph{Voronoi cell}.
	
	A vector $\bm{v}$ is called a \textit{Voronoi vector} if the hyperplane $\{\bm{x} \in \mathbb{R}^{n}: \bm{x} \cdot \bm{v}=\frac{1}{2} \bm{v} \cdot \bm{v}\}$ has a non-empty intersection with $\mathcal{V}(0).$ A Voronoi vector is said to be \textit{relevant} if this intersection is an $(n-1)-$dimensional face of $\cV(0).$  
	
	The packing radius $r_{\text{pack}}$ of a lattice $\Lambda$ is half of the minimum distance between lattice points and the packing density $\Delta(\Lambda)$ is the fraction of  space that is covered by balls $\mathcal{S}(\bm{\lambda}, r_{\text{pack}})$ of radius $r_{\text{pack}}$ in $\mathbb{R}^{n}$ centered at lattice points $\lambda \in \Lambda,$ i.e., $\Delta(\Lambda)  =  \frac{\text{vol} \ S(0,r_{\text{pack}})}{\text{vol} (\Lambda)}.$
	
	The objective of the \textit{nearest lattice point problem} is to find  $$\bm{u} = \arg \min_{\overline{u} \in \mathbb{Z}^{n}} \mid\mid \bm{x}-V\bm{\overline{u}} \mid\mid^{2},$$ where the norm considered is the standard Euclidean norm. The nearest lattice point to $\bm{x}$ is then given by $\bm{x_{nl}}=V\bm{u}$.  %The set of all $\bm{x}$ mapped to $\bm{y}\in \Lambda$ is the Voronoi region ${\mathcal V}(\bm{y})$. %Voronoi cells  partition $\mathbb{R}^n$, in the sense that their interiors are nonoverlapping. We denote the closest lattice point to $\bm{x}$ by $g_{c}(\bm{x})$.%Observe that the mapping $g_{nl}~:~\mathbb{R}^n \rightarrow \Lambda, \ \  \bm{x} \mapsto \bm{x_{nl}}$ partitions $\mathbb{R}^n$ into  Voronoi cells. %, each of volume $|\det V|

	We denote the integer and fractional parts of ${\bf x} \in \mathbb{R}$ by $\lfloor {\bf x} \rfloor $ and $\{ {\bf x}\}$, respectively. Thus ${\bf x}=\lfloor {\bf x} \rfloor  + \{{\bf x} \}$ and $0\leq \{{\bf x} \} < 1$. The nearest integer function is $[x]=\lfloor x+1/2\rfloor$.

	For a triangular generator matrix, the \textit{nearest plane (np) algorithm}~\cite{babai} computes $\bm{x_{np}}$, an approximation to $\bm{x_{nl}}$, given by $\bm{x_{np}}=V{\bm{u}} = {u}_1 \bm{v_1} + {u}_2 \bm{v_2}+\ldots+{u}_n \bm{v_n}$, where ${u}_i \in \mathbb{Z}$ is given by 
	\begin{equation}
	{u}_i=\left[ \frac{x_i-\sum_{j=i+1}^nv_{i,j}{u}_j}{v_{i,i}}\right]
	\label{eqn:babaicoeff}
	\end{equation}
	in the order $i=n,n-1,\ldots,1$. We refer to $\bm{x}_{np}$ as the Babai point for $\bm{x}$ and the closure of the set of $\bm{x}$ mapped to $\bm{y}\in \Lambda$ as the  Babai cell ${\mathcal B}(\bm{y})$.  A  method for finding $\bm{x_{np}}$  for general $V$ is in \cite{babai}. 
	
%	We note that with ties suitably broken, the union of Voronoi cells partit
	%, derived from \cite{babai}. 
	
%~Let ${\mathcal S}_i$ denote the subspace spanned by the vectors $\{\bm{v_1},\bm{v_2},\ldots,\bm{v_i}\}$, $i=1,2,\ldots,n$. Let ${\mathcal P}_i(z)$ be the orthogonal projection of $z$ onto ${\mathcal S}_i$ and let $\bm{v_{i,i-1}}={\mathcal P}_{i-1}(v_i)$ be the closest vector to $\bm{v_i}$ in ${\mathcal S}_{i-1}$. Consider the decomposition $\bm{v_i}=\bm{v_{i,i-1}}+\bm{v_{i,i-1}}^{\perp}$ and let $\bm{z_i^\perp}=\bm{z_{i}}-{\mathcal P}_{i}(z_{i})$. Start with $\bm{z_n}=\bm{x}$ and $i=n$ and compute $u_i=\left[ \langle \bm{z_{i}},\bm{v_{i,i-1}^\perp} \rangle/\|\bm{v_{i,i-1}^\perp}\|^2 \right]$, $\bm{z_{i-1}}={\mathcal P}_{i-1}(\bm{z_i})-u_i \bm{v_{i,i-1}}$, for $i=n,n-1,\ldots,1.$ Here $[x]$ denotes the nearest integer to $x.$
	
% 	The vector $\bm{u}=(u_{1}, u_{2}, \dots, u_{n})$ is the \textit{Babai point}, which is an approximate solution to the closest lattice point problem. The Babai point corresponding to $\bm{x}$ is denoted $g_{np}(\bm{x})$. %mapping $g_{np}~:~\mathbb{R}^n \rightarrow \Lambda, \ \ \bm{x} \mapsto \bm{x_{np}}$ partitions $\mathbb{R}^n$ into  hyper-rectangular cells with volume $|\det V|$ and we refer to this partition as a \emph{Babai partition}.
 	
 \begin{example} Fig.~\ref{fig-np} represents the Babai cells and the Voronoi cells (hexagons) for the hexagonal lattice $A_2$ generated by $\{(1,0),(1/2, \sqrt{3}/2)\}$ and illustrates how the \textit{np algorithm} approximates the nearest point problem. % in this lattice.  %to see why the np algorithm is an approximation to the nearest lattice point problem.

\begin{figure}[h!]
\begin{center}
		\includegraphics[scale=0.26]{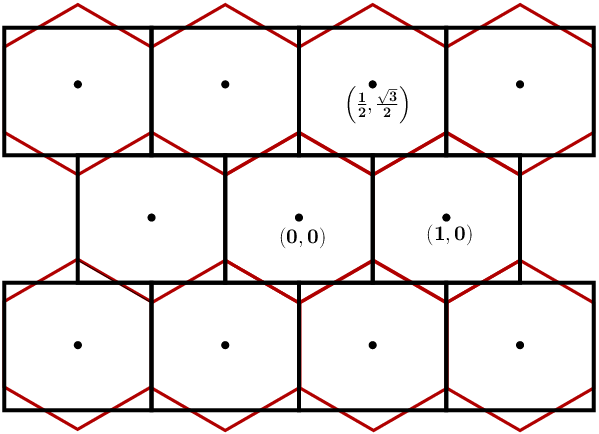}  
\caption{{Babai and Voronoi cells  for the hexagonal lattice $A_2$}} \label{fig-np}
\end{center}
\end{figure}
\end{example}

	In case the generator matrix $V$ is upper triangular with $(i,j)$ entry $v_{ij}$, each rectangular cell is axis-aligned and has sides of length  $|v_{11}|,|v_{22}|,\ldots,|v_{nn}|.$ %In this specific case, the vectors $v_{i,i-1}^{\perp}$ mentioned above are of type $(0, \ldots, v_{ii}, \ldots,0)$. 
~We remark that given a lattice $\Lambda$ with an arbitrary generator matrix $V \in \mathbb{R}^{n\times n}$ we can always apply the QR decomposition $V=QR,$ where $Q \in \mathbb{R}^{n\times n}$ is an orthogonal matrix and $R \in \mathbb{R}^{n \times n}$ is an upper triangular matrix. The matrix $R$ will then generate a rotated (and equivalent) version of the original lattice $\Lambda.$
	
\subsection{Error Probability}	
\label{sec:errprob}
We now define and simplify the error probability, $P_e$ and its complement, $P_c$, the success probability, for use in Secs.~\ref{sec4} and \ref{sec5}. The error probability and its complement are defined by $P_e=1-P_c=\prob{\bm{X}_{nl}\neq \bm{X}_{np}}$. Clearly $P_c=\sum_{\bm{y}\in \Lambda}\prob{\bm{X}_{nl}= \bm{y}, \bm{X}_{np}=\bm{y}}$. Two situations of interest are (i) $\bm{X}$ is uniformly distributed over a union of Babai cells, which we refer to hereafter as the uniform distribution case, and (ii) $\bm{x}=\bm{y}+\bm{z},$ where $\bm{y}\in \Lambda$ is the transmitted lattice vector,  and $\bm{z} \in \mathbb{R}^n$ is  Gaussian noise, $\mathcal{N}(0, \sigma^2\bm{I})$. We refer to this as the Gaussian case. %If we assume that $\prob{X_{nl}\neq \bm{y}|X_{np}=\bm{y}}$ is independent of $\bm{y}$, then $P_e=\prob{X_{nl}\neq 0|X_{np}=0}$. 

In the uniform case $$P_c=\frac{\vol{\mc{V}(0)\bigcap \mc{B}(0)}}{\vol{\mc{B}(0)}}.$$

In the Gaussian case, 

{\small \begin{eqnarray}
P_c  & =  & \sum_{\bm{y} \in \Lambda} \sum_{\bm{y}' \in \Lambda}\prob{\bm{X}_{nl}=\bm{y}',\bm{X}_{np}=\bm{y}',\bm{Y}=\bm{y}} \nonumber \\
& \stackrel{(a)}{=} & \sum_{\bm{y}\in \Lambda} \prob{\bm{Y}=\bm{y}} \sum_{\bm{y}' \in \Lambda}\prob{\bm{Z}\in \cB(\bm{y}'-\bm{y})\bigcap \cV(\bm{y}'-\bm{y})}  \nonumber \\
& = & \sum_{\bm{y}\in \Lambda} \prob{\bm{Y}=\bm{y}} \sum_{\bm{y}' \in \Lambda}\prob{\bm{Z}\in \cB(\bm{y}')\bigcap \cV(\bm{y}')} \nonumber \\
& = &  \sum_{\bm{y}' \in \Lambda}\prob{\bm{Z}\in \cB(\bm{y}')\bigcap \cV(\bm{y}')} \nonumber \\
& = &\underbrace{\prob{\bm{Z}\in \cB(\bm{0})\bigcap \cV(\bm{0})} }_{T}+ \sum_{\bm{y}' \in \Lambda, \bm{y}' \neq 0}\prob{\bm{Z}\in \cB(\bm{y}')\bigcap \cV(\bm{y}')},
\label{eqn:gaussprob}
\end{eqnarray}}
where in (a) we have asserted the independence of $\bm{Z}$ and $\bm{Y}$.
For small noise variance, the dominant term in the above sum is $T=\prob{\bm{Z} \in \cV(0) \bigcap \cB(0)}$. Note also that $P_c=1$ when the basis vectors are mutually orthogonal.

It is an important fact that the Babai cell $\cB(0)$ is dependent on the choice of the lattice basis, whereas the Voronoi cell is invariant to the choice of lattice basis. Thus, the error probability depends on the choice of basis, and in particular, the order in which the basis vectors are listed. Thus, in future sections, where we evaluate the error probability for a given  generator matrix $V$, we determine the Babai cell for all $n!$ column permutations of $V$ by applying the QR decomposition to each permutation. The error probability is then the minimum that is obtained over all column permutations.

\section{The Distributed Babai Protocol (DBP)  and its Communication Cost} %Communication Cost of the Babai Point in a Distributed Network}
\label{sec3}

%	Our first analysis is devoted to the communication cost of solving approximately the closest vector problem in a distributed network. The objective here is to estimate the cost of sending extra bits from each node of the distributed system so that the CN can recover the Babai point, given that each coordinate of the received vector $\bm{x} \in \mathbb{R}^n$ to be decoded is available at a physically separated node, as illustrated in Fig.~\ref{fig-p1}. Communication protocols are presented for the centralized model along with associated rate calculations in the limit as $\alpha \rightarrow 0.$ 
	
	We now describe the  protocol DBP, by which the Babai point $\bm{x}_{np}=V{\bm{u}}$ can be determined exactly at the computing node with a finite rate of transmission. We assume that 
\begin{enumerate}
\item the lattice $\Lambda$ has upper triangular generator matrix $V$, and
\item the ratio of any two non-zero entries in any row of ${V}$ are rational numbers.
\end{enumerate}

Define integers $p_{ml}$, $q_{ml}>0$ and relatively prime, by canceling out common factors in $v_{ml}/v_{mm}$, i.e. let $p_{ml}/q_{ml}=v_{ml}/v_{mm}$.  Let $q_m=\lcm \ \{q_{ml}, l>m\}$, where $\lcm$ denotes the least common multiple  of its arguments. By definition $q_n=1$.
The `interference' term $\nu_m$ is given by  $\nu_m={\sum_{l=m+1}^{n} u_{l}v_{ml}}/{v_{mm}}$. In terms of integer and fractional parts, $\nu_m=\lfloor \nu_m \rfloor +\{\nu_m\}$, $0 \leq \{\nu_m\} < 1$ and further, $\{\nu_m\}$ is of the form $s/q_m$, $0\leq s < q_m$. Let $\cS_m\subset\{0,1,\ldots,q_m-1\}$ be the set of values taken by $\{\nu_m\}q_m$ with positive probability. For most source probability distributions $\cS_m=\{0,1,\ldots,q_m-1\}$. However, in some cases, when $q_m$ is large this may not be the case. One such situation is described at the end of Sec.~\ref{sec:rateexamples}.
%and $p_m=\gcd \ \{p_{ml},~l>m\}$\footnote{In case $p_m \geq q_m$, then replace $p_m$ by its least residue to the modulus $q_m.$}, where $\lcm$ and $\gcd$ denote the least common multiple and greatest common divisor of its arguments respectively.  Also, note that $p_m>0$ by the definition of the $\gcd$. %Let $h_m=\gcd\ \{p_m,q_m\}$.

%We refer to $\bm{b}=(b_1,b_2,\ldots,b_n)$ obtained by the SIC decoder using 
%\begin{eqnarray}
%b_m  =  \left[ \frac{x_m-\sum_{l=m+1}^{n} b_{l}v_{ml}}{v_{mm}}\right],~~m=n,n-1,\ldots,1
%\label{eqn:BabaiCoeff}
%\end{eqnarray}
%as the \emph{SIC solution}.
\vspace{0.1in}
\noindent
{\bf Action of the Encoder in the $m$th SN}:

\noindent
Define $s_m$ %\{0,1,\ldots,q_m-1\}\bigcap p_m\mathbb{Z}$ 
to be the largest integer $s\in\cS_m$ for which 
\begin{equation}
[x_m/v_{mm}-s/q_m]=[x_m/v_{mm}]
\label{eqn:ess}
\end{equation}
Then the  $m$th SN sends 
\begin{equation}
\tilde{u}_m=[x_m/v_{mm}]
\label{eqn:btilde}
\end{equation}
 and $s_m$ to the CN in the order $m=n,n-1,\ldots,2,1$  (by definition $s(n)=0$). % i.e. $\eta_m(X_m)=(\tilde{b}_m,s_m)$.

\vspace{0.1in}
\noindent
{\bf Action of the Decoder in the CN}:

\noindent
The decoder computes $\bm{u}=(u_1,u_2,\ldots,u_n)$ where,
\begin{eqnarray} \label{decision}
\lefteqn{{u}_m=} \ \ \ \ \ \ \ \ \ \left\{ \begin{array}{cc}
\tilde{u}_m-\left\lfloor \frac{\sum_{l=m+1}^{n}{u}_{l}v_{ml}}{v_{mm}}\right \rfloor, & f_m\leq s_m,  \\
\tilde{u}_m-\left\lfloor \frac{\sum_{l=m+1}^{n}u_{l}v_{ml}}{v_{mm}}\right \rfloor-1,  & f_m >  s_m,
\end{array}
\right.
\label{eqn:dec}
\end{eqnarray}
where $\tilde{u}_m$ is given by \eqref{eqn:btilde}, 
$$f_m=\left\{  \frac{\sum_{l=m+1}^{n}{u}_{l}v_{ml}}{v_{mm}}\right\}{q_m}$$
and computation proceeds in the order $m=n,n-1,\ldots,1$.

\begin{theorem}\label{thmcost}(\textit{Decoder output is the Babai point})
The output of the decoder coincides with  the  solution $\bm{u}$ given in \eqref{eqn:babaicoeff}.
\end{theorem}
\begin{proof}
%Observe that each coefficient of $\overline{b}$ is given by
%\begin{eqnarray}
%\lefteqn{b_m=} \ \ \ \ \ \ \ \ \left[ \frac{x_m-\sum_{l=m+1}^{n} b_{l}v_{m,l}}{v_{mm}}\right],
%~~m=n,n-1,\ldots,1,
%\label{eqn:BabaiCoeff}
%\end{eqnarray}
Rewrite  \eqref{eqn:babaicoeff} in terms of   fractional and integer parts  to get
\begin{eqnarray}
u_m & = & \left[ \frac{x_m}{v_{mm}}-\left\{\frac{\sum_{l=m+1}^{n} u_{l}v_{ml}}{v_{mm}}\right\} \right]-  \left\lfloor \frac{\sum_{l=m+1}^{n} u_{l}v_{ml}}{v_{mm}}\right\rfloor, ~ ~  m=n,n-1,\ldots,1.
\label{eqn:BabaiCoeff-1}
\end{eqnarray}
The fractional part in the above equation is of the form $s/q_m$, $s \in \mathbb{Z}$ and further, $0 \leq s < q_m$.  Thus%\{0,1,\ldots,q_m-1\}\bigcap p_m\mathbb{Z}$, where $p_m,q_m$ are as  defined above, it follows that   Thus
\begin{eqnarray} \label{eqn:decision}
\lefteqn{{u}_m=} \ \ \ \ \ \ \ \ \ \left\{ \begin{array}{cc}
\tilde{u}_m-\left\lfloor \frac{\sum_{l=m+1}^{n}{u}_{l}v_{ml}}{v_{mm}}\right \rfloor, & s\leq s_m,  \\
\tilde{u}_m-\left\lfloor \frac{\sum_{l=m+1}^{n}u_{l}v_{ml}}{v_{mm}}\right \rfloor-1,  & s >  s_m,
\end{array}
\right.
\end{eqnarray}
where $\tilde{u}_m$ is given by \eqref{eqn:btilde}, and the computation  of $u_m$ is performed at the CN in the order $m=n,n-1,\ldots,1$.
\end{proof}

\subsection{Communication Cost of Protocol DBP}
%\begin{corollary}\label{corocost}
%The rate required to transmit $s_m$, $m=1,2,\ldots,n-1$ is no larger than $\sum_{i=1}^{n-1}\log_2(q_i)$ bits.
%\end{corollary}
\begin{theorem} (\textit{Sum rate of the protocol DBP})
Assume that $X_i$, $i=1,2,\ldots,n$ are mutually independent and identically distributed with known marginal probability distribution. The sum rate $R_{sum}$ of protocol DBP is
\begin{equation}
R_{sum}=\sum_{i=1}^nR_i=\sum_{i=1}^n H(\tilde{U}_i,S_i).
\label{eqn:sumrateDBP}
\end{equation}
\end{theorem}

As an example, 
suppose that $\bm{X}$ is uniformly distributed over a rectangular region $[-A/2,A/2]^n$, for $A$ large. The total rate is  
\begin{eqnarray}
R_{sum} & = & n\log_2(A)-\log_2|\det V| +\sum_{i=1}^{n-1}H(S_i|\tilde{U}_i) \nonumber \\
& \leq  & n\log_2(A)-\log_2|\det V| +\sum_{i=1}^{n-1}\log_2(q_i). 
\label{eqn:bits}
\end{eqnarray}
The first two terms in \eqref{eqn:bits} can be interpreted as the rate required to compute the Babai point for a lattice  $\Lambda' \subset \mathbb{R}^{n}$ generated by orthogonal vectors $\{(v_{11}, 0, \dots, 0), \dots, (0, 0, \dots, v_{nn})\},$ where $v_{ii}, 1 \leq i \leq n$ are the diagonal elements from the upper triangular generator matrix $V$ of the lattice $\Lambda.$ Observe that the Babai cells of $\Lambda$ are congruent to those of $\Lambda'$, but are not aligned as they are in $\Lambda'$. The last term in \eqref{eqn:bits} is the additional communication cost because of the misalignment of the Babai cells of $\Lambda$. %Voronoi region of $\Lambda'$ corresponds to the Babai partition not aligned achieved without sending any extra bit in this model. The idea is to decode in the orthogonal associated lattice $\Lambda',$ which is a simple process and after that, recover the original approximate closest lattice point in $\Lambda,$ which coincide with the cost of the extra bits. 

\subsection{Optimality of  Protocol DBP}
\label{sec:lower}

We prove optimality of the protocol DBP based on a  bound on the sum rate for the  distributed function computation problem from~\cite{Orlitsky:2001}. In order to make the derivation self-contained, we first summarize the salient facts about characteristic graphs and graph entropy which play a fundamental role in the bound derived in~\cite{Orlitsky:2001} before proceeding to derive a lower bound for protocol DBP.
Note that our bound is for continuous alphabets, and is based on a limiting form of the result stated in \cite{Orlitsky:2001}, for discrete alphabets. The limiting argument is self-evident and is not presented. 

Consider a function $f(x_1,x_2,\ldots,x_n)~:~\mathbb{R}^n \to \mathbb{Z}^n$, and our distributed computation setup where $x_i$ is available at the $i$th SN and $f$ is to be computed at the CN.  A lower bound on the communication rate from the $i$th SN to the CN is given by the minimum rate required to compute $f$, assuming that $x_j,~j\neq i$ is known at the receiver. We will use the notation $i^c=\{1\leq j \leq n,~j\neq i\}$ and $\bm{x}_{i^c}$ for the vector $(x_j,~j\neq i)$.
From \cite{Orlitsky:2001}, the minimum communication rate is given in terms of the conditional graph entropy of a specific graph. We now describe computation of the conditional graph entropy.  For convenience we will write $f(\bm{x})=f(x_i|\bm{x}_{i^c})$, when studying the communication rate from the $i$th SN to the CN, to emphasize the fact that $\bm{x}_{i^c}$ is side information at the CN.

The characteristic graph, $\cgraf_i$, of the function $f(x_i|\bm{x}_{i^c})$, has as its nodes the support of $x_i$, which in this case is $\mathbb{R}$. Two distinct nodes $x_i$ and $x'_i$ are connected by an edge if and only if (iff) there is an $\bm{x}_{i^c}$ for which $f(x_i|\bm{x}_{i^c})\neq f(x'_i|\bm{x}_{i^c})$. An independent set is a collection of nodes, no two of which are connected by an edge.  A maximal independent set is an independent set which is not contained in any other independent set. The minimum rate required to compute $f_i(x_i|\bm{x}_{i^c})$ with $\bm{x}_{i^c}$ known at the CN is given by the conditional graph entropy $H_{\cgraf_i}(X_i|\bm{X}_{i^c})$~\cite{Orlitsky:2001}, described next. Let $\Gamma_i$ be the collection of  maximal independent sets of $\cgraf_i$ and let $W$ be a random variable which takes the values $w\in \Gamma_i$---thus the realizations of $W$ are maximally independent sets.  Let $p(w|x_i,\bm{x}_{i^c})$ be a conditional probability distribution with the following properties:
\begin{enumerate}
\item \label{item:1.1} $p(w|x_i,\bm{x}_{i^c})=p(w|x_i)$, for all $w\in \Gamma_i$,$(x_i,\bm{x}_{i^c}) \in \mathbb{R}^n$. (Markov condition).
\item \label{item:1.2} $p(w|x_i)=0$ if $x_i \notin w$.
\item \label{item:1.3} $\sum_{w\in \Gamma_i} p(w|x_i) = 1$.
\end{enumerate}
Let $\cP_i$ be the collection of all such probability distributions. Then by definition
\begin{equation}
H_{\cgraf_i}(X_i|\bm{X}_{i^c})= \min_{p \in \cP_i}I(W;X_i|\bm{X}_{i^c}).
\label{eqn:cgent}
\end{equation}

%We will see that for our problem, the maximal independent sets are intervals. A given $x_i\in \mathbb{R}$, will in general belong to more than one maximal independent set. Here we will see that this is not the case. Each $x_i$ belongs to exactly one maximal independent set. 

%In our case, we wish to compute $\bm{u}(\bm{x})$ given by \eqref{eqn:babaicoeff}. To determine a lower bound on the rate $R_i$ from the $i$th SN to the CN, we consider the function $\bm{u}(x_i|\bm{x}_{i^c})$. From the triangular structure of  $\bm{u}(\bm{x})$ it follows that  $$\bm{u}(x_i|\bm{x}_{i^c})=\left[ \frac{x_i-\sum_{j=i+1}^nv_{i,j}u_j}{v_{i,i}}\right],$$ for $i=n,n-1,\ldots,1$. % where we have used the notation $i^c$ to denote the set $\{j~:1\leq j \leq n, j\neq i\}$ and $x_S$ is the sub-vector of $(x_1,x_2,\ldots,x_n)$, indexed by the set $S$ (thus $x_{\{1,3,4\}}=(x_1,x_3,x_4)$). Our objective is to determine the structure of $\cgraf_i$, the characteristic graph for the function $u_i$ since it known~\cite{Orlitsky:2001} that
%$R_{sum} \geq \sum_{i=1}^nH_{\cgraf_i}(X_i|X_{i^c})$,  where $H_{\cgraf_i}(X_i|X_{i^c})$ is the conditional graph-entropy of $\cgraf_i$. 

%The characteristic graph  $\cgraf_i$ is defined  on the support set of $X_i$. Each node is an element of the support set  of $X_i$ and $x_i$, $x'_i$ are disconnected iff  $u_i(x_i|x_{i^c})=u_i(x'_i|x_{i^c})$ for all $x_{i^c}$. 
We now apply this machinery for obtaining a lower bound on the rate $R_i$ for computing 
$$\bm{u}(x_i|\bm{x}_{i^c})=\left[ \frac{x_i-\sum_{j=i+1}^nv_{i,j}u_j}{v_{i,i}}\right],$$ for $i=n,n-1,\ldots,1.$
Our goal is to determine $\cgraf_i$ and its maximal independent sets, $i=1,2,\ldots,n$, and the probability distribution that solves \eqref{eqn:cgent}.

First consider $\cgraf_n$. In $\cgraf_n$, $x_n$ is \emph{disconnected} from $x'_n$  iff $[x_n/v_{n,n}]=[x'_n/v_{n,n}]$ or equivalently the maximal independent sets are the level sets of $[x_n/v_{n,n}]$. Since $x_n$ lies in exactly one of these sets, it follows from item~\ref{item:1.2}  and \eqref{eqn:btilde} that  $W=\tilde{U}_n$.  Hence $R_n \geq \min_{p\in \cP_n} I(W;X_n|\bm{X}_{n^c})= H(\tilde{U}_n|\bm{X}_{n^c})$, since $H(\tilde{U}_n|X_n)=0$.

 Now consider $\cgraf_m$ for $m<n$. As before, let $\nu=\sum_{j=m+1}^n v_{m,j}u_j/v_{m,m}$ and write $\nu = \{\nu\}+\lfloor \nu \rfloor$. Since $\{\nu\}=s/q_m$, $s\in \cS \subset \{0,1,\ldots,q_m-1\}$ it follows that $x_m$ and $x'_m$ are disconnected in $\cgraf_m$ iff  $[x_m/v_{m,m}-s/q_m]=[x'_m/v_{m,m}-s/q_m]$  for all $s \in \cS \subset \{0,1,\ldots,q_m-1\}$ or equivalently, $[x_m/v_{m,m}]=[x'_m/v_{m,m}]$ and the value of $s_m$ evaluated using \eqref{eqn:ess} is the same for $x_m$ and $x'_m$. %lies in an interval of the form $\left[(l-1/2+s/q_m)v_{m,m},(l-1/2+s/q_m+1/q_m)v_{m,m}\right)$, for some $s\in \{0,1,\ldots,q_m-1\}$ and these are the maximal independent sets of $\cgraf_m$.
% Thus, the maximal independent sets of $\cgraf_m$, $m<n$, are intervals in one-one correspondence with  $(\tilde{u}_m,s_m)$, given by  \eqref{eqn:btilde} and \eqref{eqn:ess}, respectively and 
 From item~\ref{item:1.2} and \eqref{eqn:btilde}, it follows that $W=(\tilde{U}_m,S_m)$ and hence $R_m \geq \min_{p\in \cP_m} I(W;X_m|\bm{X}_{m^c})= H(\tilde{U}_m,S_m|\bm{X}_{m^c})$.
 %and uniquely identifies a pair $(i,j) \in \mathbb{Z}\times\{0,1,\ldots,q_m-1\}$.

% Suppose that $x_m-x'_m > 1/q_m$. We show that there is a value of $\nu=\sum_{j=m+1}^n v_{m,j}u_j/v_{m,m}$ such that  $u_i(x_i|x_{i^c})\neq u_i(x'_i|x_{i^c})$. Indeed, let $l$ be the largest integer such that 
%$$
%\left[ \frac{x_m}{v_{m,m}}-\frac{l}{q_m}\right]=\left[ \frac{x_m}{v_{m,m}}\right]
%$$
%and set $\nu=l/q_m$. Then $u_i(x'_i|x_{i^c})=u_i(x_i|x_{i^c})-1$. 
%
%Let $\cZ_m$ denote the collection of maximally independent sets of $\cgraf_m$ and let $Z_m$ be a random variable which takes values in $\cZ_m$. From~\cite{Orlitsky:2001} we have  $H_{\cgraf_m}(X_m|X_{m^c})$, the conditional graph-entropy of $\cgraf_m$ is by definition the minimum conditional mutual information  $I(Z_m;X_m|X_{m^c})$, where the minimum is over all conditional distributions $P(Z_m|X_m,X_{m^c})$ which satisfy 
%\begin{enumerate}
%\item the Markov condition $P(Z_m|X_m,X_{m^c})=P(Z_m|X_m)$,
%\item $\sum_{z\in \cZ_m~:~x_m\in z}P(z|x_m)=1$.
%\end{enumerate}
%Since $Z_m$ is uniquely determined by $X_m$ it follows that $H(Z_m|X_m)=0$ and hence$$H_{\cgraf_m}(X_m|X_{m^c})=\min I(X_m;Z_m|X_{m^c})=H(Z_m|X_{m^c}).$$ 
% Thus the sum rate satisfies the lower bound
Thus (recall that $S_n=0$)
\begin{equation}
R_{sum} = \sum_{i=1}^nR_i \geq \sum_{i=1}^n H(\tilde{U}_i,S_i|\bm{X}_{i^c}).
\end{equation}
Since the lower bound coincides with the sum rate of the protocol DBP given by \eqref{eqn:sumrateDBP} when the $X_i$ are mutually independent, DBP is optimal.

\subsection{Examples}
\label{sec:rateexamples}
\begin{figure}[htbp] %  figure placement: here, top, bottom, or page
   \centering
   \includegraphics[width=3in]{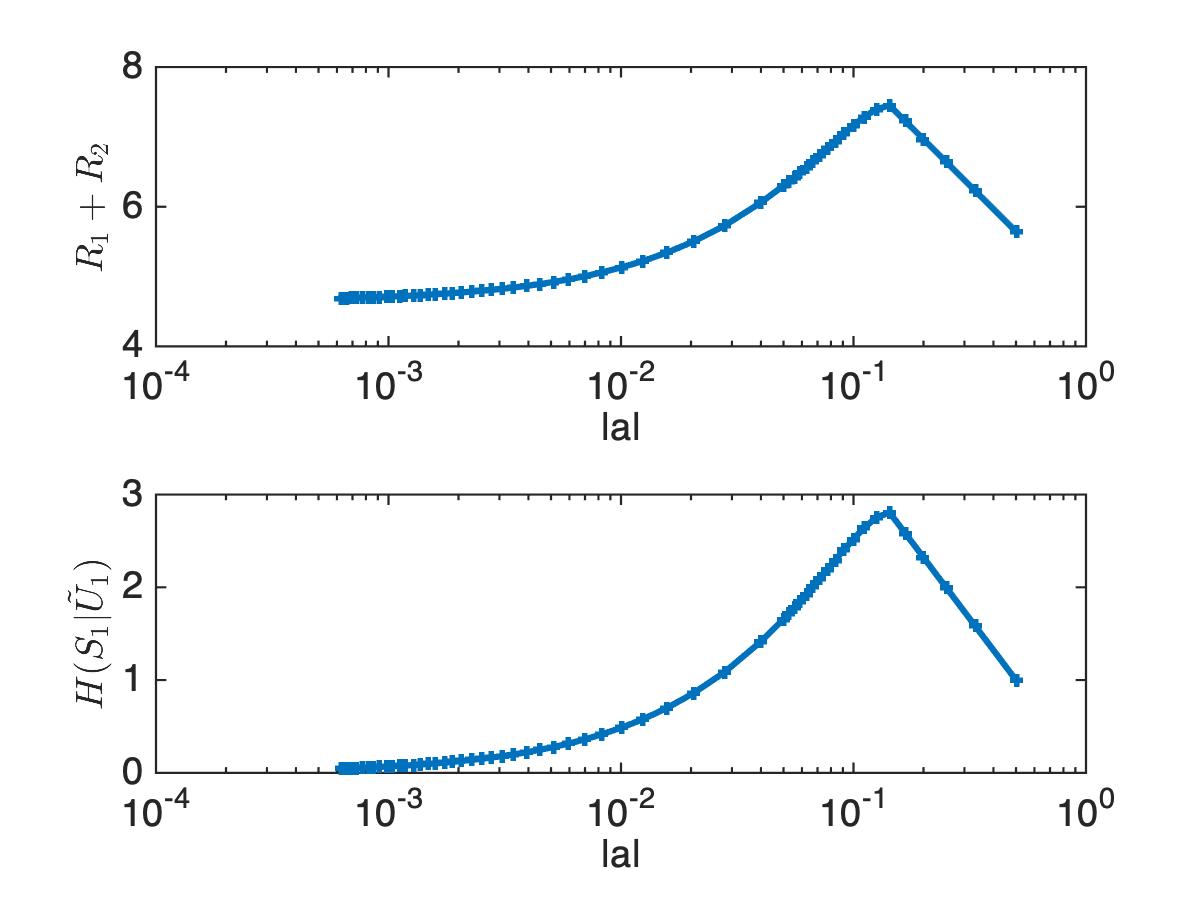} 
   \caption{Communication rates for 2 dimensional lattices and a uniform source distribution over the square $[-5/2,5/2)\times [-5/2,5/2)$. The basis vectors are $(1,0)$ and $(a,b)=(1/m,\sqrt{1-1/m^2})$,  with integer $m\geq 2$.}
   \label{fig:latticerates}
\end{figure}

In 	the following examples, we  illustrate how the method proposed in Th. \ref{thmcost} works,  present a case where the communication cost is  large, and compute communication rates for  a family of two-dimensional lattices, for a uniformly distributed source.
	\begin{example} Consider the three dimensional body-centered cubic (BCC) lattice with basis $\left\{(1,0,0),\right.$ $\left. (-\tfrac{1}{3}, \tfrac{2\sqrt{2}}{3}, 0), (-\tfrac{1}{3}, -\tfrac{\sqrt{2}}{3}, \sqrt{\tfrac{2}{3}}))\right\}.$
The Babai point given by $\bm{u}=(u_1, u_2, u_3),$ is given by
 \begin{eqnarray} 
u_{3} = \left[ \sqrt{\frac{3}{2}} x_{3} \right], ~ ~ ~ u_{2}=\left[ \frac{3}{2\sqrt{2}}x_{2} + \left\{ \frac{1}{2} u_{3} \right\} \right]+\left\lfloor\frac{1}{2} u_{3} \right\rfloor, \nonumber
\end{eqnarray}
\begin{eqnarray*}
\text{and}~ ~ ~ u_{1}&=&\left[ x_1 + \left\{ \frac{1}{3} u_{2} + \frac{1}{3} u_{3} \right\} \right]+\left\lfloor\frac{1}{3} u_{2} + \frac{1}{3} u_{3} \right\rfloor.
%&=& \left[ \frac{3}{2\sqrt{2}}x_{2} + \left\{\frac{1}{2}\left[ \sqrt{\frac{3}{2}} x_{3} \right] \right\} \right] +\left\lfloor\frac{1}{2} \left[ \sqrt{\frac{3}{2}} x_{3} \right] \right\rfloor, 
\end{eqnarray*}
In order for the Babai point ${\bf u}$ to be correctly calculated at the  CN, nodes $2$ and $1$ send the following extra information, according to the protocol DBP:
\begin{align*}
\text{node 2:} \ \left\{ \frac{1}{2} u_{3} \right\} = \frac{s_2}{q_2}, \ q_2=2 \ \text{then} \ s_2=0 \ \text{or} \ 1 \\
\text{node 1:} \ \left\{ \frac{1}{3} u_{2} + \frac{1}{3} u_{3} \right\} = \frac{s_1}{q_1}, \ q_1=3 \ \text{then} \ s_1=0,1 \ \text{or} \ 2 .
\end{align*}
	Observe that the values of $s_1$ and $s_2$ are calculated for a general received vector ${\bf x}=(x_1, x_2).$ Therefore, the sum rate to send $s_1$ and $s_2$ to the CN is $\log_{2}2 + \log_{2}3 \approx 2.5859 \approx 3 \ \text{bits}.$
\end{example}

\begin{example} Consider a two-dimensional lattice with basis $\{(1,0),(\tfrac{311}{1000}, \tfrac{101}{100}) \}.$ We have that
 \begin{equation}
u_{2}=\left[\frac{x_{2}}{v_{22}} \right]=\left[\frac{100}{101}x_{2} \right]
\end{equation} 
and
\begin{eqnarray*}
u_{1}&=&\left[\frac{x_{1}}{v_{11}} - \left\{\frac{u_{2}v_{21}}{v_{11}} \right\} \right]-\left\lfloor\frac{u_{2}v_{21}}{v_{11}} \right\rfloor = \left[ x_{1} -  \left\{\left[\frac{100}{101}x_{2} \right] \frac{311}{1000} \right\} \right] - \left\lfloor\left[ \frac{100}{101}x_{2} \right] \frac{311}{1000} \right\rfloor.
\end{eqnarray*} 
Consider, for example, $x=(1, 1),$ then $\left\{\left[\frac{100}{101}x_{2} \right] \frac{311}{1000} \right\}=\frac{311}{1000}=\frac{s}{q}.$ In this case, node $1$ must send the largest integer $s_1$ in the range $\{0,1, \dots, 999\}$ for which $\left[x_{1}-\frac{s_1}{q_{1}}\right]=[x_{1}]$ and we get $s_1=500.$ This procedure will cost no larger than $\log_{2}q_{1}=\log_{2}1000 \approx 9.96$ and in the worst case, we need to send almost $10$ bits to recover the Babai point at the CN. 

\end{example}

Communication rates for various two-dimensional lattices are presented in Fig.~\ref{fig:latticerates} for a source uniformly distributed over the square $[-5/2,5/2)\times[-5/2,5/2)$. The basis vectors are $(1,0)$ and $(a,b)$, $a^2+b^2=1$, with $a=1/m$, and integer $m\geq 2$.  The sum rate is seen to peak at $a=1/6$.  Consider the case where $m=991$. Note that $u_2=[x_2/b]$ and $u_1=[x_1-au_2]$. The scaled fractional interference term  $m\{au_2\}$ takes values in $\cS=\{0,1,2,3,988,989,990\}$ which is a much smaller set than $\{0,1,\ldots,990\}$. This observation is essential for ensuring that the conditional entropy $H(S_1|\tilde{U}_1)$ eventually decreases as $a\to 0$.
%The analysis here points to the importance of the number-theoretic structure of the generator matrix  $V$ in determining the communication requirements for computing the Babai $\bm{u}$ (and respectively $\bm{x_{np}} = V\bm{u}.$) 

%\textcolor{red}{I am not sure about this last paragraph!}

\section{Error Probability Calculations for Dimensions $n=2,~3$:}
\label{sec4}
We have presented a protocol for computing the  Babai point in a distributed network and evaluated its communication  cost.  We now explore several issues related to the Babai point. 

First, since the Babai point is an approximation for the nearest lattice point, it is of interest to evaluate the probability that the two points are unequal, i.e., the error probability $P_e$ as defined in Sec.~\ref{sec:errprob}. In this section we analyze $P_e$ for the uniform case.  The Gaussian case is presented in a later section. %An example to demonstrate such dependence is now presented.
Efficient numerical computation of $P_e$ requires that we work with special bases as defined in Sec.~\ref{secMOS}. Analytic and numerical computation of $P_e$ for $n=2,3$ is then addressed in Secs.~\ref{sec-twod} and \ref{secthreed}. Knowledge of the error probability is useful because in some situations it might be sufficient to compute the Babai point, and not incur the extra communication cost of finding the nearest lattice point. We mention here that the additional cost of finding the true nearest lattice point has been addressed in dimension two in~\cite{VB:2017}.

Second, we study the variation of the error probability $P_e$ with the packing density of the lattice. The intuition driving this study is that as the packing density increases, the Voronoi cell become increasingly spherical, and we should expect the error probability to increase. We see that some well-known regular polyhedra lie on the optimal tradeoff curve between the packing density and the error probability. Numerical evidence about the nature of polyhedra that lie on this optimal tradeoff curve is also presented.
	
%	Since the Babai point and the nearest lattice point are not identical for every $\bm{x}\in \mathbb{R}^n$,  it is of interest to evaluate the error probability, i.e. the probability that the two points are unequal. This section is then devoted to analyzing the error probability for two and three dimensional lattices, assuming that the received random vector $\bm{x} \in \mathbb{R}^n$ is uniformly distributed over the Babai partition. In Sec.~\ref{secMOS} we make some general observations about dependence of the error probability on the basis, and then prove a theorem which is useful in developing the algorithm for computing the error probability, especially for $n=3$. Calculations for $n=2$ and $n=3$ are then presented in Secs.~\ref{sec-twod} and~\ref{secthreed}, respectively.

\subsection{Special Bases: Minkowski and Obtuse Superbase}
\label{secMOS}

	 A basis $\{\bm{v_{1}},\bm{v_{2}},...,\bm{v_{n}}\}$ of a lattice $\Lambda \subset \mathbb{R}^{n}$ is said to be \textit{Minkowski-reduced}
if $\bm{v_{j}},$ $j=1,\dots,n,$ is such that $\left\Vert \bm{v_{j}}\right\Vert \leq\left\Vert \bm{v}\right\Vert $, for any $\bm{v}$ such that $\{\bm{v_{1}},...,\bm{v_{j-1}},\bm{v}\}$ can be extended
to a basis of $\Lambda$. %An analogous condition for a basis $\{\bm{v_{1}}, \bm{v_{2}}\}$ to be Minkowski-reduced in dimension two is $||\bm{v_{1}}||\leq||\bm{v_{2}}||\leq||\bm{v_{1}}\pm \bm{v_{2}}||.$ We present below a result about relevant vectors derived from \cite{ConwaySloane:1992}.

\begin{theorem}\cite{conwaysloane} \label{propmink} (\textit{Minkowski-reduced basis from Gram matrix}) Consider the Gram matrix $A$ of a lattice $\Lambda.$ The inequalities from Eq.~(\ref{mink1}), Eqs.~(\ref{mink1})--(\ref{mink2}) and Eqs.~(\ref{mink1})--(\ref{mink3}) below define a Minkowski-reduced basis for dimensions 1,2 and 3, respectively.
\begin{eqnarray}
0 < a_{11}  & \leq & a_{22} ~ \leq ~ a_{33}  \label{mink1} \\ %\leq \dots \leq a_{nn} \label{mink1} \\
2|a_{st}| & \leq & a_{ss} \ \ (s < t) \label{mink2} \\
2|a_{rs} \pm a_{rt} \pm a_{st}| & \leq & a_{rr} + a_{ss} \ \ (r < s < t). \label{mink3}
\end{eqnarray}
\end{theorem}
All lattices in $\mathbb{R}^n$ have a Minkowski-reduced basis, which  roughly speaking, consists of short vectors that are as  perpendicular as possible~\cite{conwaysloane}. In dimension two, relevant vectors can be determined from a Minkowski-reduced basis as follows.

\begin{lemma}\cite{ConwaySloane:1992} \label{lemmathird}(\textit{Relevant vectors given a Minkowski-reduced basis}) Consider a Minkowski-reduced basis of the form $\{(1,0),(a,b)\}$ and let $\theta$ be the angle between $(1,0)$ and $(a,b)$. Then besides the basis vectors, a third relevant vector is
\begin{equation}
\begin{cases}
(-1+a,b), & \text{if } \frac{\pi}{3} \leq \theta \leq \frac{\pi}{2} \\
(1+a,b), & \text{if } \frac{\pi}{2} < \theta \leq \frac{2\pi}{3}.
\end{cases}
\end{equation}
\end{lemma}

	In dimension two, the characterization~\cite{conwaysloane} for a Minkowski-reduced basis is the following: %, p.397{]}:
a lattice basis $\left\{ \ensuremath{\bm{v_{1}},\bm{v_{2}}}\right\} $
is Minkowski-reduced if only if $\left\Vert \bm{v_{1}}\right\Vert \leq\left\Vert \bm{v_{2}}\right\Vert $
and $2 |\bm{v_{1}} \cdot \bm{v_{2}}|  \leq \left\Vert \bm{v_{1}}\right\Vert ^{2}.$ 
Consequently, the angle $\theta$ between $\bm{v_{1}}$ and $\bm{v_{2}}$ is such that $\text{ }\frac{\pi}{3} \leq\theta\leq\frac{2\pi}{3}.$ 

We describe next the concept of an obtuse superbase that will be applied in the three-dimensional approach.

	Let $\{\bm{v_1},\bm{v_{2}}, \dots, \bm{v_{n}}\}$ be a basis for a lattice $\Lambda \subset \mathbb{R}^n$. A {\textit{superbase}}   $\{\bm{v_{0}}, \bm{v_{1}}, \dots, \bm{v_{n}}\}$ with $\bm{v_0}= -\sum_{i=1}^{n} \bm{v_{i}},$  is said to be {\textit{obtuse}} if $p_{ij}=\bm{v_{i}} \cdot \bm{v_{j}} \leq 0,$ for $i,j=0,\dots,n, \ \ i \neq j$. A lattice $\Lambda$ is said to be of {\textit Voronoi's first kind} if it has an \textit{obtuse superbase}. The existence of an obtuse superbase allows a characterization of the relevant Voronoi vectors of a lattice \cite[Th.3, Sec. 2]{ConwaySloane:1992}, which are of the form $\sum_{i \in S} \bm{v_{i}},$ where $S$ is a strict non-empty subset of $\{0,1,\dots, n\}.$

%\begin{theorem} \label{obtusesuper} \cite[Th.3, Sec. 2]{ConwaySloane:1992} (\textit{Voronoi vectors given an obtuse superbase}) Let $\Lambda$ be a lattice of Voronoi's first kind with obtuse superbase $\{\bm{v_0},\bm{v_1}, \dots, \bm{v_n}\}$. %The Voronoi vectors of $\Lambda$ are of the 
%Vectors of the form $\sum_{i \in S} \bm{v_{i}},$ where $S$ is a strict non-empty subset of $\{0,1,\dots, n\}$ are Voronoi vectors of $\Lambda$.
%\end{theorem}

	It was demonstrated \cite{ConwaySloane:1992} that all lattices with dimension less or equal than three are of Voronoi's first kind and given the existence of obtuse superbases for three dimensional lattices, their Voronoi regions can be classified into five possible parallelohedra which we  present in the sequel. 
	
	Given an obtuse superbase, since ${\bf v_0}=-{\bf v_1}-{\bf v_2}-{\bf v_3},$ all Voronoi vectors can be  written as one of the following seven vectors or their negatives:
\begin{equation*}
{\bf v_{1}},{\bf v_{2}},{\bf v_{3}},{\bf v_{12}}={\bf v_{1}}+{\bf v_{2}}, {\bf v_{13}}={\bf v_{1}}+{\bf v_{3}}, {\bf v_{23}}={\bf v_{2}}+{\bf v_{3}},{\bf v_{123}}={\bf v_{1}}+{\bf v_{2}}+{\bf v_{3}}.
\end{equation*}
The Euclidean norm of such vectors  $N(v_{1}),N(v_{2}),N(v_{3}),N(v_{12}),N(v_{13}),$ $N(v_{23}), N(v_{123})$ are called \textit{vonorms} and $p_{ij}=-{\bf v_{i}} \cdot {\bf v_{j}}~(0 \leq i < j \leq 3)$ are denoted as \textit{conorms}. % for the superbase $\{{\bf v_0},{\bf v_1}, {\bf v_2},{\bf v_3}\}.$ Precise definitions of conorms and vonorms for the general $n-$dimensional case can be found in \cite{ConwaySloane:1992}.	

\begin{remark} \label{remark5type}
	The Voronoi region of a lattice $\Lambda \subset \mathbb{R}^n$ with obtuse superbase $\{{\bf v_0},{\bf v_1}, {\bf v_2},{\bf v_3}\}$ can be classified \cite{ConwaySloane:1992} according to the five choices of zeros for their conorms, which leads to five possible parallelohedra, as presented in Fig.~\ref{5types}. 
	The characterization is based on the conorms as follows:
\begin{itemize}
\item cuboid, if $p_{12} = p_{13} = p_{23} = 0.$
\item hexagonal prism, if only two conorms among $p_{12}, p_{13}$ and $p_{23}$ are zero.
\item rhombic dodecahedron, if only one $p_{12}, p_{13}$ or $p_{23}=0$ and $p_{0j}$ are nonzero for all $j=1,2,3.$
\item hexa-rhombic dodecahedron, if only one $p_{12}, p_{13}$ or $p_{23}=0$ and $p_{0j}=0,$ for $j=1,2,3.$
\item truncated octahedron, if all $p_{ij} ~(0 \leq i < j \leq 3)$ are nonzero.
\end{itemize}	
\end{remark}

%	The nonzero cosets of $\Lambda/2\Lambda$ naturally form a discrete projective plane of order two. The vonorms %$p_{i|jkl}=N(v_{i}),p_{ij|kl}=N(v_{ij})= N(v_{kl}),$ where $\{i,j,k,l\}$ is any permutation of $\{0,1,2,3\},$ 
%are marked as the nodes of the projective plane and the corresponding conorms $0$ and $p_{ij}$ at the nodes of the dual plane in the following Fig.~\ref{gpp}.
%	
%\begin{figure}[H]
%\begin{center}
%		\includegraphics[height=4cm]{projectiveplaneMM}  
%\caption{{Projective and dual planes labelled with vonorms and conorms respectively (based on \cite{ConwaySloane:1992}, p. 61)}}
% \label{gpp}
%\end{center}
%\end{figure}	
%
%	From the dual projective plane, i.e., the one constructed from the conorms, the Voronoi region of a lattice $\Lambda$ with obtuse superbase $\{{\bf v_0},{\bf v_1}, {\bf v_2},{\bf v_3}\}$ can be classified according to the five choices for zeros, which leads to five possible parallelohedra: truncated octahedron, hexa-rhombic dodecahedron, rhombic dodecahedron, hexagonal prism and cuboid, according to Fig.~\ref{5types}.
	%: one, two, three collinear zeros, three non-collinear zeros or four zeros. Each of these configuration produces a different Voronoi cell according to Figure \ref{5types}.
	
\begin{figure}[H]
\begin{center}
		\includegraphics[height=3.2cm]{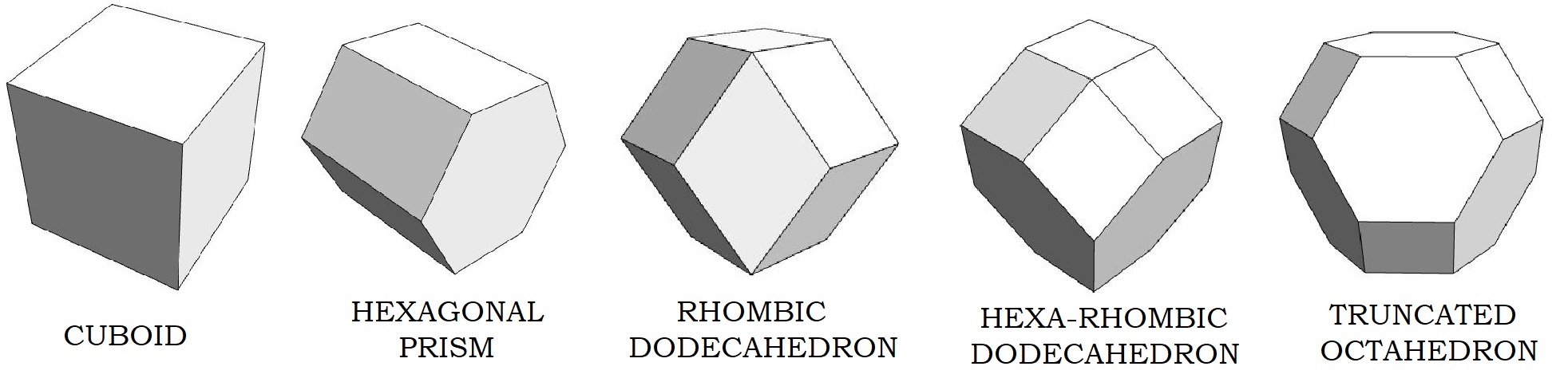}  
\caption{{The five possible shapes for a Voronoi cell of a three-dimensional lattice%(based on \cite{ConwaySloane:1992}, p. 65).
}}
 \label{5types}
\end{center}
\end{figure}

%	into five possible parallelohedra: truncated octahedron, hexa-rhombic dodecahedron, rhombic dodecahedron, hexagonal prism and cuboid.
	
	Now that the Minkowski-reduced basis and obtuse superbase have been defined, we present a relation between them.
	
\begin{theorem}\label{minkos} (\textit{Minkowski-reduced basis and obtuse superbase})
In dimensions $n=1,2,3$, if a lattice $\Lambda \subset \mathbb{R}^n$ has a  Minkowski-reduced basis $\{\bm{v_1},\ldots,\bm{v_n}\}$, where $\bm{v_i} . \bm{v_j} \leq 0$, $i\neq j$, then the superbase $\{\bm{v_0},\bm{v_1},\ldots,\bm{v_n}\}$ is an obtuse superbase for $\Lambda$. Conversely, if $\Lambda$ has an obtuse superbase, then a Minkowski-reduced basis can be constructed from it.
\end{theorem}

\begin{proof}
The case $n=1$ is trivial, hence we will start with $n=2.$ 

\noindent $(\Rightarrow)$ Suppose that $\{\bm{v_{1}},\bm{v_{2}}\}$ is a Minkowski-reduced basis, then, according to Th. \ref{propmink}, $0< \bm{v_{1}}\cdot \bm{v_1} \leq \bm{v_2} \cdot \bm{v_2}$ and $2|\bm{v_1} \cdot \bm{v_2}| \leq \bm{v_{1}} \cdot \bm{v_1}.$ Moreover, by hypothesis, $\bm{v_1} \cdot \bm{v_2} \leq 0.$ Define $\bm{v_0}=-\bm{v_1}-\bm{v_2}$ and to guarantee that $\{\bm{v_0},\bm{v_1},\bm{v_2}\}$ is an obtuse superbase, we need to check that $p_{01} \leq 0$ and $p_{02} \leq 0.$ Indeed,
$
p_{01}= \bm{v_0} \cdot \bm{v_1} = (-\bm{v_1} -\bm{v_2}) \cdot \bm{v_1} = -\bm{v_1} \cdot \bm{v_1} \underbrace{-\bm{v_1} \cdot \bm{v_2}}_{|\bm{v_1} \cdot \bm{v_2}|} \leq -2|\bm{v_1} \cdot \bm{v_2}| + |\bm{v_1} \cdot \bm{v_2}| \leq 0.
$
Similarly we have that $p_{02} \leq 0.$
	
\noindent $(\Leftarrow)$ If $\{\bm{v_0},\bm{v_1},\bm{v_2}\}$ is an obtuse superbase, any permutation of it is also an obtuse superbase. So, we may consider one such that $|\bm{v_1}| \leq |\bm{v_2}| \leq |\bm{v_0}|.$ Then we have that $0 < \bm{v_1} \cdot \bm{v_1} \leq \bm{v_{2}} \cdot \bm{v_{2}} \leq (\bm{v_1}+\bm{v_2}) \cdot (\bm{v_1}+\bm{v_2})$ and $\bm{v_{1}} \neq 0.$ From the last inequality, we have that $ -2\bm{v_{1}} \cdot \bm{v_{2}} \leq  \bm{v_{1}} \cdot \bm{v_{1}} \Rightarrow 2 |\bm{v_{1}} \cdot \bm{v_{2}}| \leq  \bm{v_{1}} \cdot \bm{v_{1}}.$\\

\noindent For n=3:
 $(\Rightarrow)$ Consider a Minkowski-reduced basis $\{\bm{v_1},\bm{v_{2}},\bm{v_{3}}\}$ such that $\bm{v_{1}}\cdot \bm{v_{2}} \leq 0, \bm{v_{1}} \cdot \bm{v_{3}} \leq 0$ and $\bm{v_{2}} \cdot \bm{v_{3}} \leq 0.$ To check if $\{\bm{v_{0}},\bm{v_1},\bm{v_2},\bm{v_{3}}\}$ is an obtuse superbase, we need to verify that $p_{01} \leq 0, p_{02} \leq 0$ and $p_{03} \leq 0.$ One can observe that
\[
p_{01}= \bm{v_0} \cdot \bm{v_{1}} = -\bm{v_1} \cdot \bm{v_1} \underbrace{-\bm{v_1} \cdot \bm{v_2}}_{|\bm{v_{1}} \cdot \bm{v_{2}}|} \underbrace{-\bm{v_1} \cdot \bm{v_3}}_{|\bm{v_1} \cdot \bm{v_3}|} \leq  -\bm{v_1} \cdot \bm{v_1} +  \frac{\bm{v_1} \cdot \bm{v_1}}{2} + \frac{\bm{v_1} \cdot \bm{v_1}}{2} \leq 0.
\]
With analogous arguments, we show that $p_{02} \leq 0$ and $p_{03} \leq 0.$
	
\noindent $(\Leftarrow)$ To prove the converse, up to a permutation, we may consider an obtuse superbase such that $|\bm{v_{1}}| \leq |\bm{v_2}| \leq |\bm{v_3}| \leq |\bm{v_{0}}|.$ This basis will be Minkowski-reduced if we prove conditions ($\ref{mink2}$) and ($\ref{mink3}$) from Th. \ref{propmink}, i.e.,
\begin{equation} \label{ine1}
2|\bm{v_1} \cdot \bm{v_2}| \leq \bm{v_1} \cdot \bm{v_1}; \ \  2|\bm{v_1} \cdot \bm{v_3}| \leq \bm{v_1} \cdot \bm{v_1}; \ \  2|\bm{v_2} \cdot \bm{v_3}| \leq \bm{v_2} \cdot \bm{v_2},
\end{equation}
\begin{equation} \label{ine2}
2|\pm \bm{v_1} \cdot \bm{v_2} \pm \bm{v_{1}} \cdot \bm{v_{3}} \pm \bm{v_{2}} \cdot \bm{v_{3}}| \leq \bm{v_1} \cdot \bm{v_1}+ \bm{v_2} \cdot \bm{v_2}.
\end{equation}

	The inequalities in Eq. \eqref{ine1} are shown similarly to the two dimensional case starting from $\bm{v_{2}} \cdot \bm{v_{2}} \leq (\bm{v_1}+\bm{v_2}) \cdot (\bm{v_1}+\bm{v_2}),$ $\bm{v_{3}} \cdot \bm{v_{3}} \leq (\bm{v_1}+\bm{v_3}) \cdot (\bm{v_1}+\bm{v_3})$ and $\bm{v_{3}} \cdot \bm{v_{3}} \leq (\bm{v_2}+\bm{v_3}) \cdot (\bm{v_2}+\bm{v_3}).$ Starting from $\bm{v_{3}} \cdot \bm{v_{3}} \leq (\bm{v_1} + \bm{v_2} +\bm{v_3}) \cdot (\bm{v_1} + \bm{v_2} + \bm{v_3}),$ the inequality in Eq. \eqref{ine2} follows, concluding the proof.  
\end{proof}

	Characteristics of Voronoi vectors of low-dimensional lattices can be found in \cite{May95}. For our application, the obtuse superbase (\cite[Th.3, Sec. 2]{ConwaySloane:1992}) leads to considerable simplification in identifying all the relevant vectors for a Voronoi cell. For more details about low dimensional reduced bases, see \cite{nguyen04}. Computation of a Minkowski-reduced basis in high dimensions is a hard problem and the basis commonly used in practice is an approximation, obtained using the  the LLL algorithm\cite{lll}. %In high dimensions, the basis commonly used is the LLL \cite{lll}, which coincides with the Minkowski-reduced basis in dimension $2.$ % however due to the geometric aspect of our problem, we would consider only the two bases previously mentioned in this section.

%	Also with respect to Minkowski-reduced basis, there is an interesting result that characterizes the strict Voronoi vectors of a lattice for low dimensional lattices.
%	
%\begin{theorem} \cite{May95} Let $\Lambda \subset \mathbb{R}^n$ be a lattice of dimension $n<7.$ Let $\{\bm{v_1}, \dots, \bm{v_n}\}$ be a Minkowski reduced basis of $\Lambda.$ Then each $\bm{v_i},$ where $i=1, \dots, n,$ is a strict Voronoi vector.
%\end{theorem}

\subsection{Error Probability and Packing Density: Two-dimensional lattices, Uniform Distribution}
\label{sec-twod}
	We consider that a Minkowski-reduced lattice basis, which is also obtuse (Th. \ref{minkos}) can be chosen by the designer of the lattice code and it can be transformed into an equivalent basis $\{(1,0),(a,b)\},$ by applying QR decomposition to the lattice generator matrix. % in addition to convenient scalar factor. Note that, under the assumption of uniform probability distributions, the error probability estimation is not affected by a scalar factor. %The reason for working with a Minkowski-reduced basis is justified by Example~\ref{ex1} and Lemma~\ref{lemmathird}.
	
%	Note that, if $\{{\bm{v\ensuremath{_{1}},\bm{v_{2}}}}\}$ is a Minkowski basis then so is $\{{-\bm{v\ensuremath{_{1}},\bm{v_{2}}}}\}$ and hence any lattice has a Minkowski basis with $\frac{\pi}{2}\leq\theta\leq\frac{2\pi}{3}$. So, if we 
	From the Minkowski-reduced basis $\{(1,0),(a,b)\},$ where $a^{2}+b^{2} \geq 1$ and $-\frac{1}{2} \leq a \leq 0,$ it is possible to use Lem.~\ref{lemmathird} to describe the Voronoi region of $\Lambda$ and determine its intersection with the associated Babai cell. Observe that the area of both regions must be the same and in this specific case, equal to $|b|.$ %This means that the vertices that define the Babai rectangular partition are $\left(\pm \frac{1}{2}, \pm \frac{b}{2}\right).$ 

	In addition $\{(-1-a,-b),(1,0),(a,b)\}$ is an obtuse superbase for $\Lambda,$ so the relevant vectors that defines the Voronoi region are $\pm (1,0), \pm (a,b)$ and $\pm (-1-a,-b).$ We will choose for the analysis proposed in Thm.~\ref{thmmain} only the vectors in the first quadrant, i.e., $(1,0),(1+a,b), (a,b),$ due to the symmetry of the Voronoi cell. Hence, the following result states a closed formula for the error probability $P_e:=\prob{\bm{X}_{np} \neq \bm{X}_{nl}}$ %= \mathcal{P}(\bm{u} \neq \bm{u})$ 
of any two-dimensional lattice.
	
\begin{theorem} \cite{BVC:2017} \label{thmmain} (\textit{Error probability for two-dimensional lattices}) Consider a lattice $\Lambda \subset\mathbb{R}^{2}$ with a Minkowski-reduced basis $\{\bm{v_1},\bm{v_2}\}=\{(1,0),(a,b)\},$ such that the angle $\theta$ between $\bm{v_{1}}$ and $\bm{v_{2}}$ satisfies \textup{$\frac{\pi}{2}\leq\theta\leq \frac{2\pi}{3}$.} The error probability $P_e,$ when the received vector $\bm{x}=(x_1, x_2) \in \mathbb{R}^2$ is uniformly distributed over the Babai cell, is 
\begin{equation}
P_e=F(a, b)=\frac{-a-a^{2}}{4b^{2}}=\frac{1-(1+2a)^{2}}{16b^{2}}.
\end{equation} 
\end{theorem}

\begin{proof} We are going to present just the main idea of the proof. A detailed version is available in \cite[Thm. 1]{BVC:2017}. According to Lemma \ref{lemmathird}, Fig. \ref{rvectors}, we can find the vertices of the Voronoi cell $\cV(0)$, which are: $\pm(\frac{1}{2},\frac{a^{2}+b^{2}+a}{2b})$,
$\pm(-\frac{1}{2},\frac{a^{2}+b^{2}+a}{2b})$ and $\pm(\frac{2a+1}{2},\frac{-a^{2}+b^{2}-a}{2}),$ while the Babai cell $\cB(0)$ has vertices $(\pm \frac{1}{2}, \pm \frac{b}{2}).$ 

\begin{figure}[H]
\begin{center}
		\includegraphics[scale=0.26]{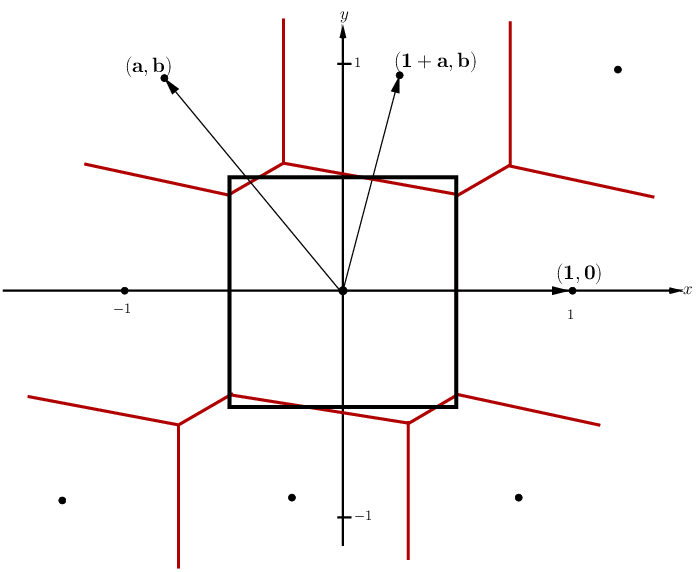}  
\caption{{Voronoi cell, Babai cell and three relevant vectors}}
\label{rvectors}
\end{center}
\end{figure}

From Fig. \ref{rvectors}, the error probability is calculated as the sum of the areas of four `error' triangles normalized by the area of a Babai cell. The explicit formula for it is $F(a, b) =  \dfrac{1}{4} \dfrac{-a-a^{2}}{b^2}.$
\end{proof}	

%$P_{e}$ is then computed as the ratio between the area of the Babai
%region which is not overlapped by the Voronoi region $\mathcal{V}(0)$ and the area $|b|$ of the Babai
%region. From Fig. \ref{rvectors}, we get the error as the sum the areas of four triangles, where two of them are defined respectively by the points $\left(\frac{1}{2},\frac{b}{2}\right), \left(\frac{1}{2},\frac{a^{2}+a+b^{2}}{2b}\right), \left(\frac{a+1}{2}, \frac{b}{2}\right)$ and $\left(-\frac{1}{2},\frac{b}{2}\right), \left(-\frac{1}{2},\frac{a^{2}+a+b^{2}}{2b}\right), \left(\frac{a}{2}, \frac{b}{2}\right).$ The remaining two triangles are symmetric to these two. Therefore, the error probability is the sum of the four areas, normalized by the area of the Voronoi region $|\det(V)|=|b|$. The explicit formula for it is given by $F(a, b) =  \dfrac{1}{4} \dfrac{-a-a^{2}}{b^2}.$
%& = & \frac{1}{\rho\sin\theta} \frac{1}{4} \frac{\cos\theta}{\sin\theta}(1-\rho\cos\theta)\\

%	In a sequel, we state the result that provides a tight upper and lower bound for the error probability  $P_e:= \mathcal{P}(\bm{x_{np}} \neq \bm{x_{nl}})$  of any two-dimensional lattice. 

%	, illustrated in Fig. \ref{contour}, from  the error probability $P_{e}=F(a,b)=\frac{1}{4} \frac{a}{b^2}(1-a)=\frac{1-(1+2a)^{2}}{16b^{2}}$ obtained in Theorem \ref{thmmain} with  $b \geq \frac{\sqrt{3}}{2}$ and $-\frac{1}{2} \leq a \leq 0.$
		
\begin{corollary} \label{corope} (\textit{Error probability analysis for two dimensional lattices}) For any two dimensional lattice with a Minkowski-reduced basis satisfying the conditions of  Thm.~\ref{thmmain}, we have 
\begin{equation}
0 \leq P_{e} \leq \frac{1}{12},
\end{equation}
and
\begin{itemize}
\item[a)] $P_{e}=0 \Longleftrightarrow a=0,$ i.e., the lattice is orthogonal.
\item[b)] $P_{e}=\frac{1}{12} \Longleftrightarrow (a,b)=\left(-\frac{1}{2}, \frac{\sqrt{3}}{2}\right),$ i.e., the lattice is equivalent to the hexagonal lattice.
\item[c)] the level curves of $P_{e}$ are described as ellipsoidal arcs (Fig.~\ref{contour}) in the region $a^{2}+b^{2} \geq 1$ and $-\frac{1}{2} \leq a \leq 0$ (condition required for the basis to be Minkowski-reduced).
\end{itemize}
\end{corollary}

\begin{figure}[h!]
\begin{center}
		\includegraphics[height=6cm]{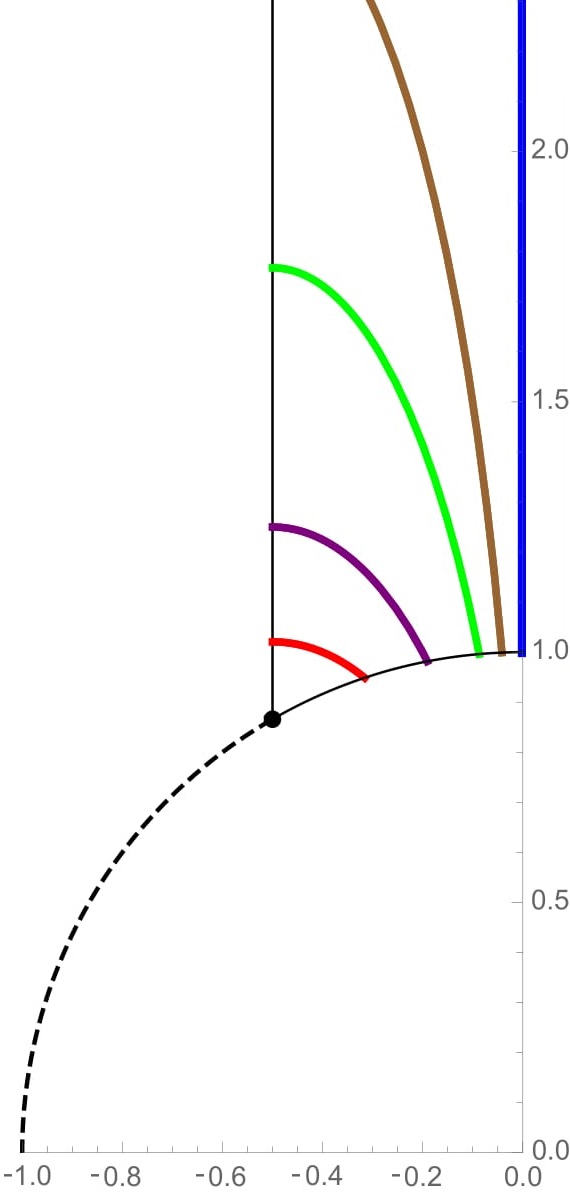}  
\caption{{Level curves of $P_{e}=k,$ in right-left ordering, for $k=0, k=0.01, k=0.02, k=0.04, k=0.06$ and $k=1/12 \approx 0.0833.$ Notice that $a$ is represented in the horizontal axis and $b$ in vertical axis. }}
\label{contour}
%\vspace{0.15cm} {\footnotesize {Fontes: \cite{slidesueli}, p.4 e \cite{zamir}, p.18.} }
\end{center}
\end{figure}

\begin{remark}\label{remark_ep_pd} From Corollary \ref{corope},	one can notice a straightforward relation between the packing density of the lattice and its error probability. The packing density of a lattice with basis $\{(1,0),(a,b)\}$ is given by $\Delta_2(a,b)=\tfrac{\pi}{4b}$ and $F(a,\Delta_2)=\frac{{\Delta^{2}_2}[1-(1+2a)^{2}]}{\pi^{2}},$ following the notation from Th. \ref{thmmain}. For a fixed $a,$ the error probability increases with $\Delta_2,$ and for a fixed density $\Delta_2$ and fixed $b,$ the error probability is decreasing with $a,$ where $-\tfrac{1}{2} \leq a \leq \min\left\{  -\sqrt{1-\left( \frac{\pi}{4\Delta_{2}}\right)^2},0\right\}.$
	
	Indeed, if we consider the error probability for a given density $\Delta_2$, we have that $F(a,\Delta_2)$ is minimized by $a=a^*$, where
\begin{eqnarray}
a^*=\left\{ \begin{array}{cc} 0, &  \Delta_2 \leq \frac{\pi}{4} \ \ \ \  (b^2 \geq 1) \nonumber \\
                       -\sqrt{1-\left( \frac{\pi}{4\Delta_{2}}\right)^2},  & \frac{\pi}{4} < \Delta_2 \leq \frac{\pi}{2\sqrt{3}} \  \ \ \ (3/4 \leq b^2 < 1). \end{array} \right.
\end{eqnarray}
and maximized by $a= -\frac{1}{2},$ for any $\Delta_2.$ Fig.~\ref{fig:packden2} represents the minimum error probability function $F(a, \Delta_2)$ for $\frac{\pi}{4} \leq \Delta_2 \leq \frac{\pi}{2\sqrt{3}}$ and expresses how the error probability varies with the packing density $\Delta_2.$

\begin{figure}[h!] %  figure placement: here, top, bottom, or page
   \centering
   \includegraphics[width=3.0in]{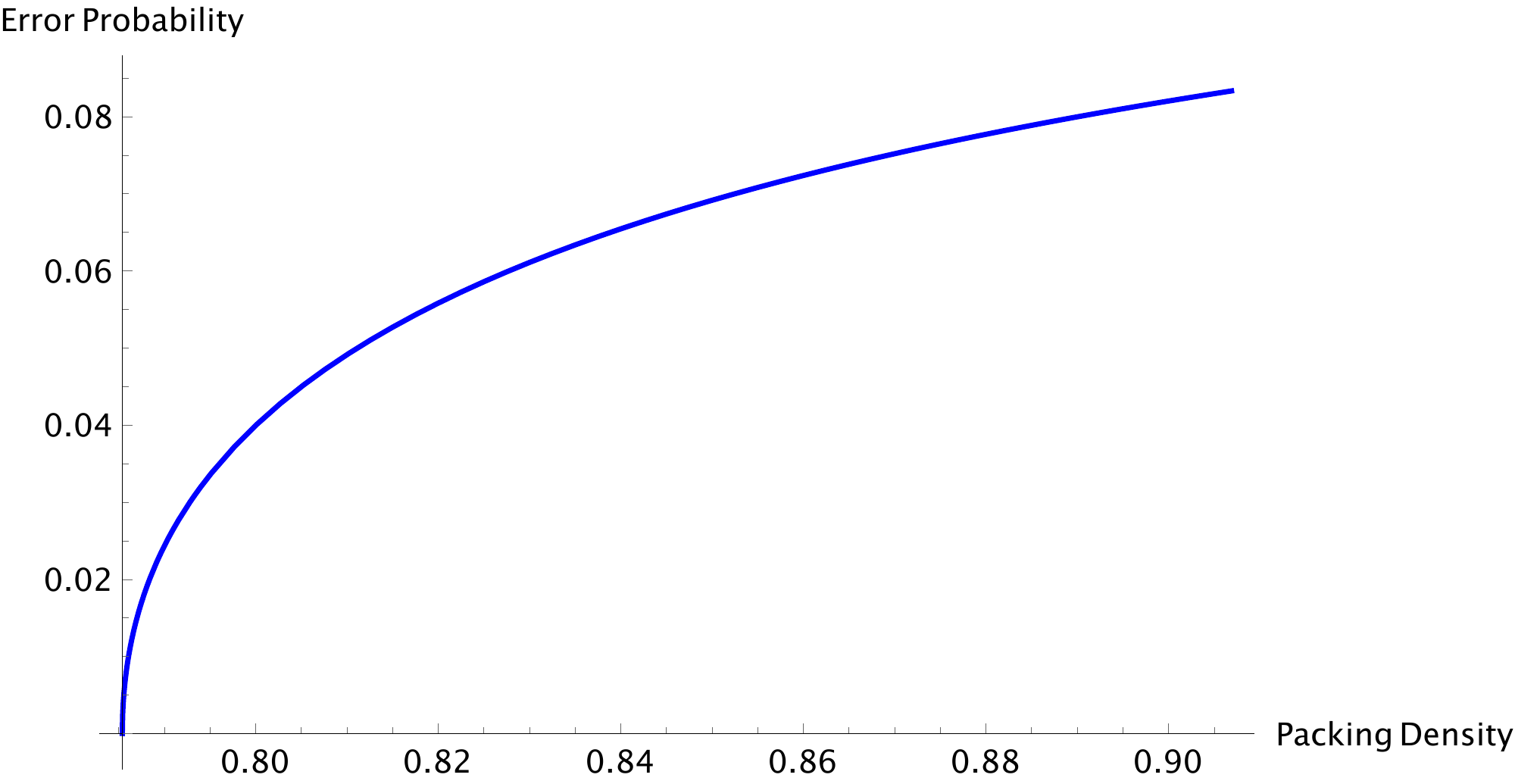} 
   \caption{Minimum error probability for given packing density assuming $\frac{\pi}{4} < \Delta_2 \leq \frac{\pi}{2\sqrt{3}},$ considering a uniform distribution}
   \label{fig:packden2}
\end{figure}
\end{remark}

%	 Note also that %starting from any Minkowski-reduced basis of a two-dimensional lattice $\gamma=\{v_{1},v_{2}\},$ 
%if $\rho:={\left\Vert v_{2}\right\Vert }$
%and  $\theta$ is defined to be the angle between the basis vectors, then the result of Th. \ref{thmmain} can be rewritten as 
%\begin{equation}
%P_{e}=H(\theta,\rho)= \frac{1-(1+2\rho \cos^{2} \theta)^{2}}{16\rho^{2} \sin^{2} \theta}.
%\end{equation}
%In this case, we can see that for a fixed $\rho,$ the error probability increases with $\theta,$ achieving its minimum in $\theta=0$ and maximum in $\theta= \frac{\pi}{2} + \arccos \frac{1}{2\rho}.$

\subsection{Error Probability and Packing Density: Three-dimensional lattices,  Uniform Distribution}
\label{secthreed}
	For the three dimensional case, we developed and implemented an algorithm in the software \textit{Wolfram Mathematica, version $12.1$}~\cite{Wolfram} which calculates the error probability of any three dimensional lattice, given an obtuse superbase, by following the characterization given in \cite{ConwaySloane:1992}. We assume an initial upper triangular lattice basis given by $\{(1,0,0),(a,b,0),(c,d,e)\},$ where $a,b,c,d,e \in \mathbb{R}.$ %which can be accomplished by performing a QR decomposition and a multiplication by a scalar factor in the original basis.
	
	It is important to remark that in dimensions greater than two, the error probability is dependent on the basis ordering. Hence, in order to analyze the smallest error probability for a given lattice, we relax the ordering imposed for the Minkowski-reduced basis and allow any permutation of a basis from now on. Our algorithm searches over all orderings and determines the best one. As an example, the performance of the BCC lattice is invariant over basis ordering, due to its symmetries. On the other hand, for the FCC lattice, depending on how the basis is ordered, we can find two different error probabilities, $0.1505$ and $0.1667,$ but we choose to tabulate the smallest one. A detailed description of the algorithm is presented in Alg.~\ref{alg}. 

\begin{algorithm}[h!] 
\caption{Error probability and packing density computation, $n=3$,  for basis $\{(1,0,0),(a,b,0),(c,d,e)\}$.} %of the nearest lattice point problem in a distributed system (three dimensional case)}
\label{alg}
\begin{algorithmic}
\vspace{0.2cm}

%\noindent\rule{\textwidth}{0.6pt}
%{\textbf{Algorithm 1}} Error probability of the nearest lattice point problem in a distributed system (three dimensional case)\\
%\noindent\rule{\textwidth}{0.6pt}			

\item[] {\bf Voronoi cell:} Given an obtuse superbase, determine the vertices of the  Voronoi cell $\cV(0)$ of $\Lambda$  %by the method proposed by Conway and Sloane \cite{ConwaySloane:1992}, 
using the  Voronoi vectors (Sec.~\ref{secMOS}). Use \textit{ConvexHullMesh[]} available in  Mathematica~\cite{Wolfram} to obtain the convex hull of the vertices of $\cV(0)$. \\%From the vertices of the Voronoi region, the function already , called , generates the Voronoi region by calculating the convex hull of those vertices.\\

\item[] {\bf Babai cell:} Determine the vertices of the Babai cell $\cB(0)$. %Since the basis has the form, $\{(1,0,0),(a,b,0),(c,d,e)\}$, these vertices are simply $\left(\pm \frac{1}{2}, \pm \frac{b}{2}, \pm\frac{e}{2} \right).$ 
Apply function \textit{ConvexHullMesh[]} to compute the convex hull of these vertices. \\ %generates the Babai region (which is an hyperrectangle) by calculating .\\

\item[] {\bf Intersection:} Apply \textit{RegionIntersection[]} in Mathematica~\cite{Wolfram},  to compute $\cB(0)\bigcap \cV(0)$ and  its volume normalized by the volume of the lattice.\\
	
\item[] {\bf Packing density:} Calculate the packing density $\Delta_3=\frac{\pi}{6} \frac{d_{\min}^3(\Lambda)}{\text{vol} (\Lambda)}.$

\end{algorithmic}
\end{algorithm}
%\noindent\rule{\textwidth}{0.6pt}

	For lattices with randomly chosen basis, we start by considering a basis at random, with the format $\{(1,0,0),$ $(a,b,0),(c,d,e)\},$ where $a,c \in [-1/2,0]$ and $b,d,e \in [-2,2]$ (the choice of the range is justified because we are only interested in lattices whose packing density is greater than $0.4$). Then, the program tests if this basis is an obtuse superbase. If this condition is false, another random basis is generated until a suitable one is found. At the end of this stage, we will have a randomly chosen obtuse and Minkowski-reduced superbase for the lattice $\Lambda.$ 
	
	Fig. \ref{fig-hp} has points given by known lattices, together with random points (orange) that are associated with lattices having a packing density greater than $0.4.$  Note that with overwhelming probability, all orange points with a randomly chosen basis have a truncated octahedron as Voronoi region, which is the most general Voronoi region in three dimensions.
\vspace{-1cm}
\begin{figure}[h]
\centering
\begin{minipage}{.35\linewidth}
  \includegraphics[height=4cm]{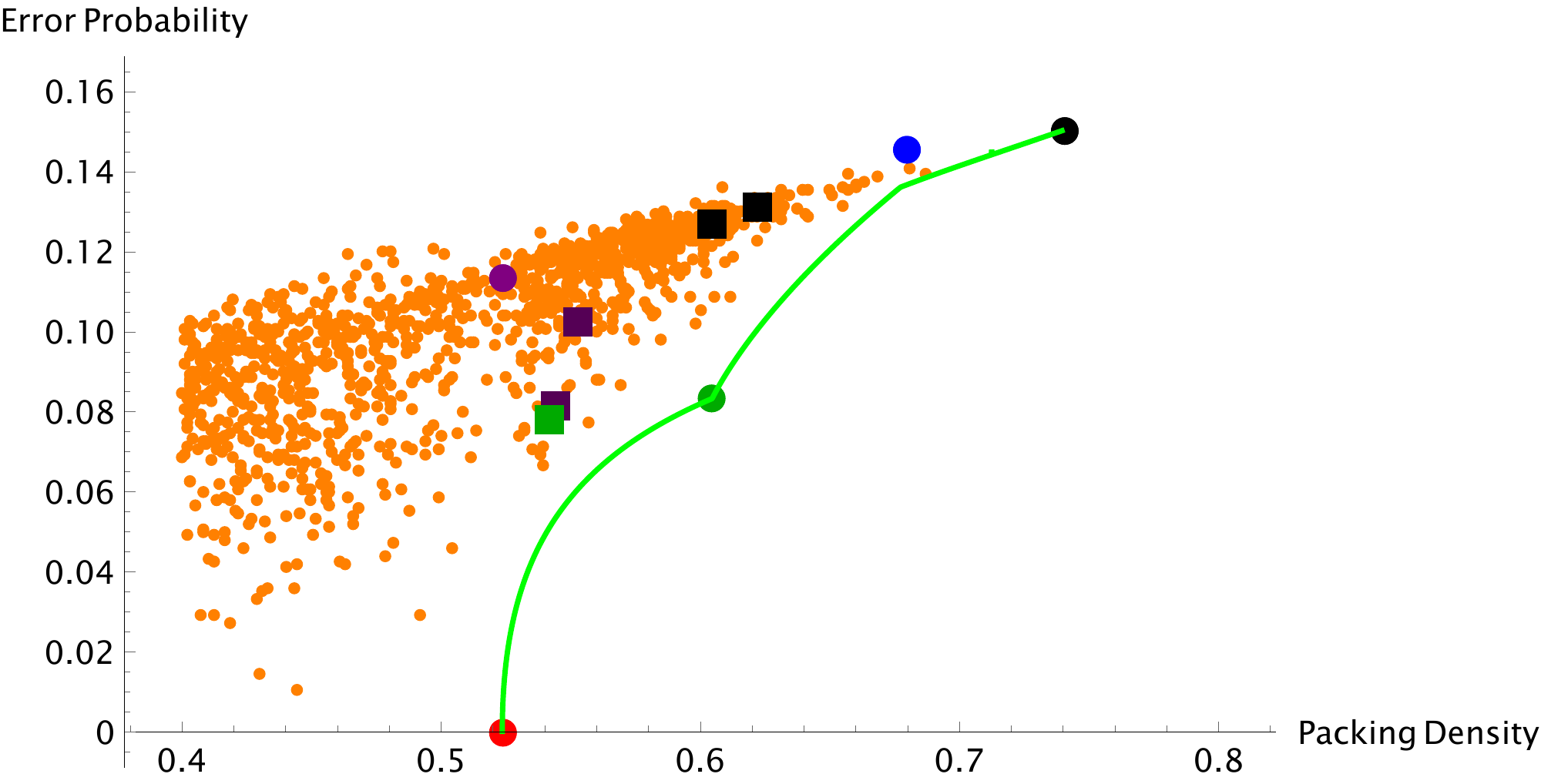}
\end{minipage}
\hspace{.15\linewidth}
\begin{minipage}{.35\linewidth}
  \includegraphics[width=\linewidth]{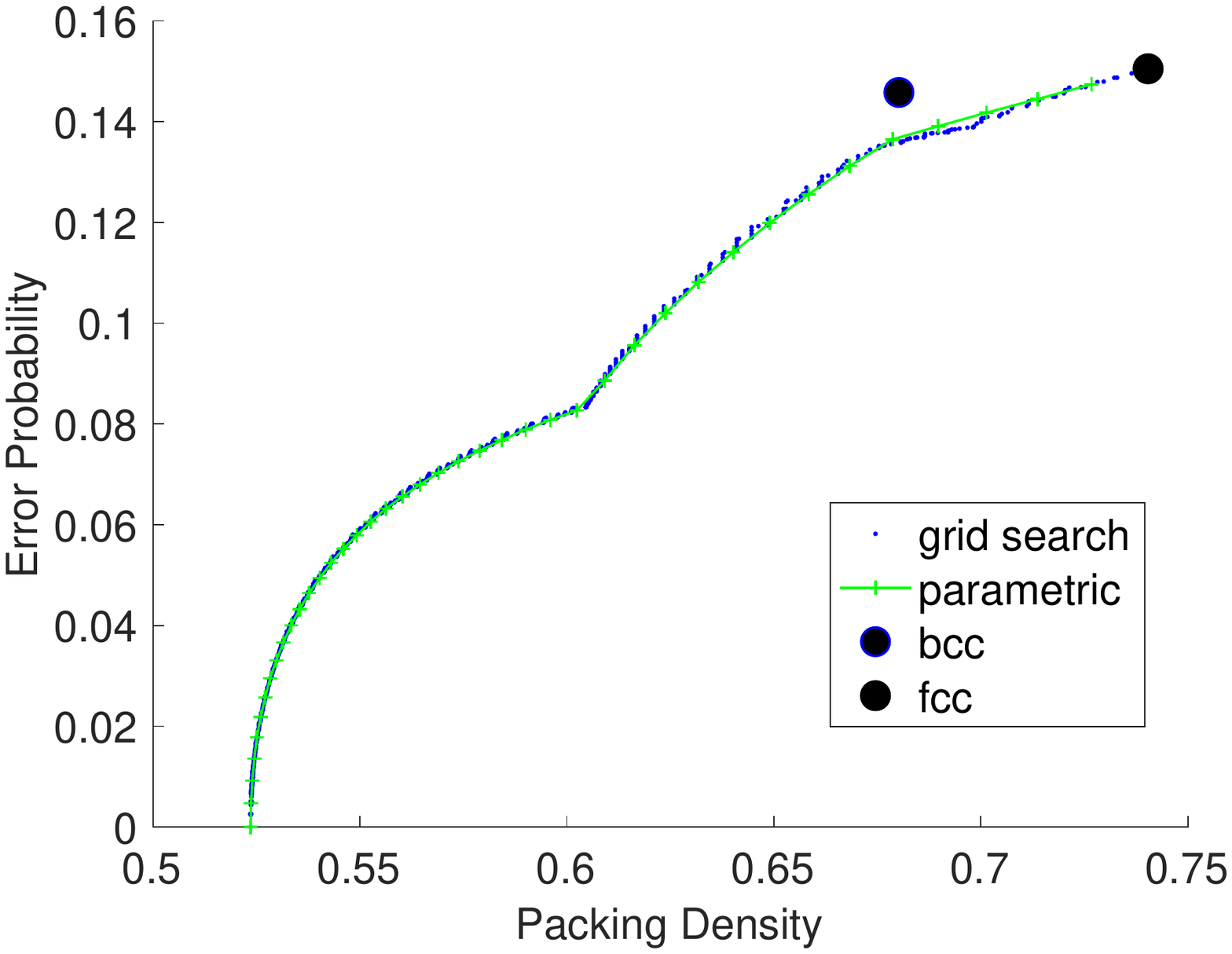}
\end{minipage}
\vspace{-1cm}
\caption{Plot of error probability and packing density for $n=3$, (left) known and randomly chosen (orange points) lattices,  (right) best points obtained from a grid search and the parametric representation.}
\label{fig-hp}
\end{figure}

%\begin{figure}[H]
%\hfill
%\centering
%\subfigure[Random (orange) and known lattices \label{fig-plot3}]{\includegraphics[width=3in]{Plot_Random1.pdf}}
%\hfill
%\subfigure[Random (orange) lattices and `hexagonal prisms' \label{fig-hp}] {\includegraphics[width=3in]{Plot_Random2.pdf}}
%\caption{Comparison between random and known performances}
%\end{figure}	

	The circular points in Fig.~\ref{fig-hp} are respectively described as:~in \textbf{\textcolor{red}{red}}, the cubic lattice $\mathbb{Z}^{3}$ with basis $\{(1,0,0),(0,1,0),(0,0,1)\};$ in \textbf{\textcolor{green}{green}}, the lattice  $\Lambda_{hp}$ with basis $\{(1,0,0),$ $(-\frac{1}{2}, -\frac{\sqrt{3}}{2}, 0),(0, 0, 1)\},$ whose Voronoi region is a regular hexagonal prism; in \textbf{\textcolor{blue}{blue}}, the body-centered cubic lattice with basis $\{(1,0,0),(-\frac{1}{3},\frac{2 \sqrt{2}}{3},0),$ $(-\frac{1}{3},-\frac{\sqrt{2}}{3},\sqrt{\frac{2}{3}})\}$, whose Voronoi region is a truncated octahedron; in \textbf{black}, the face-centered cubic lattice with basis $\{(1,0,0),(-\frac{1}{2},-\frac{1}{2},\frac{1}{\sqrt{2}}),$ $(0,1,0)\},$ whose Voronoi region is a rhombic dodecahedron; in \textbf{\textcolor{purple}{purple}}, the lattice $\Lambda_{hrd}$ with basis $\{(1,0,0),$ $(-\frac{1}{\sqrt{5}}, \frac{2}{\sqrt{5}},0),(0,-\frac{1}{2}, \frac{\sqrt{5}}{2})\},$ whose Voronoi region is a hexa-rhombic dodecahedron. Table \ref{lp} summarizes their performances when we run Alg.~\ref{alg}. 		
	
%	We also identify cases that are `almost' like one of the degenerate polyhedra, illustrated as square points in Fig. \ref{fig-hp}, where the color corresponds to the cell type, i.e., \textbf{\textcolor{green}{green}} is an hexagonal prism, \textbf{\textcolor{purple}{purple}} are hexa-rhombic dodecahedrons, and \textbf{black} represents rhombic dodecahedrons. 
	
%\begin{itemize}
%\item in \textbf{\textcolor{red}{red}}, the cubic lattice $\mathbb{Z}^{3}$ with basis $\{(1,0,0),(0,1,0),(0,0,1)\};$
%\item in \textbf{\textcolor{green}{green}}, the lattice with basis $\{(1,0,0),$ $(-\frac{1}{2}, -\frac{\sqrt{3}}{2}, 0),(0, 0, 1)\},$ Voronoi region:  hexagonal prism;
%\item in \textbf{\textcolor{blue}{blue}}, the body-centered cubic (BCC) lattice, with basis $\{(1,0,0),(-\frac{1}{3},\frac{2 \sqrt{2}}{3},0),$ $(-\frac{1}{3},-\frac{\sqrt{2}}{3},\sqrt{\frac{2}{3}})\}$, Voronoi region: truncated octahedron;
%\item in \textbf{black}, the face-centered cubic (FCC) lattice, with basis $\{(1,0,0),(0,1,0),$ $(-\frac{1}{2},-\frac{1}{2},\frac{1}{\sqrt{2}})\}$; Voronoi region: rhombic dodecahedron;
%\item in \textbf{\textcolor{purple}{purple}}, lattice  with basis $\{(1,0,0),$ $(-\frac{1}{2}, -\frac{\sqrt{5}}{2}, 0),(0, \frac{1}{\sqrt{5}}, \frac{2}{\sqrt{5}})\},$  Voronoi region: hexa-rhombic dodecahedron.
%\end{itemize} 

{\small \begin{table}[h!] 
\begin{minipage}{\columnwidth}
\caption{Performance in Algorithm~\ref{alg} for known lattices} \label{lp}
\centering
 \begin{tabular}{ | c | c | c | c | }  
  \hline\noalign{\smallskip} 
 % \cline{1-2}\noalign{\smallskip}
  %Lattice & Notation & %Conorms $(\alpha, \beta, \gamma,a,b,c)$ & 
  %$\Delta_{3}$ & $P_{e}$  \\
   %Voronoi cell& $15.6$ \cite{conwaysloane} & & \\
   Lattice/Voronoi cell & Notation $15.6$~\cite{conwaysloane} & $\Delta_{3}$ & $P_{e}$\\
  \noalign{\smallskip}\hline\noalign{\smallskip}
  $\mathbb{Z}^3$/ Cuboid & $1_{}1_{}1$& %%$(-1,-1,-1,0,0,0)$ & 
  $0.5235$ & $0$ \\
   $\Lambda_{hp}$/ Hexagonal prism & $2_{-1}2_{}2$ & %$(-\frac{1}{2},-\frac{1}{2},-1, 0,0,-\frac{1}{2})$  & 
   $0.6046$ & $0.0833$ \\
   FCC/ Rhombic dodecahedron & $2_{1}2_{1}2$ & %$(-\frac{1}{2},-\frac{1}{2},0,-\frac{1}{2},-\frac{1}{2},0)$ & 
   $0.7404$ &  $0.1505$   \\
  $\Lambda_{hrd}$/ Hexa-rhombic dodecahedron & $2_{1}3_{1}2$ & %$(-\frac{1}{2},-\frac{1}{2},-\frac{1}{2},0,-\frac{1}{2},-\frac{1}{2})$  & 
  $0.5235$ & $0.1134$ \\
   %Hexa-rhombic dodecahedron & $2_{1}2_{1}3$ & & $0.5597$ & $0.1383$ \\
   BCC/ Truncated octahedron & $3_{1}3_{1}3_{-1}$ & %$(-1,-,1-,1,-1,-1,-1)$ & 
   $0.6802$ & $0.1458$  \\
   \noalign{\smallskip}\hline
\end{tabular}
\end{minipage}
\end{table}	 }

	Fig.~\ref{fig-hp} also presents some particular cases (square points), where the numerical random search led to a Voronoi region different than the general truncated octahedron. The color corresponds to the cell type, i.e., \textbf{\textcolor{green}{green}} is an hexagonal prism, \textbf{\textcolor{purple}{purple}} are hexa-rhombic dodecahedrons, and \textbf{black} represents rhombic dodecahedrons.

\subsection{Some Observations and Analysis of the Data} 

%%%%%Sueli%%%%%%%%%%

%For three-dimensional lattices, a lower bound on the error probability can be
%obtained in terms of two-dimensional lattices. Given a general lattice $\Lambda$ with obtuse Minkowski-reduced basis $\{(1, 0, 0), (a, b, 0),(c, d, e)\},$ a lower bound for $P_e(\Lambda)$, considering this basis ordering, is given by the
%error probability of the associated prism lattice $\Lambda_p$ with basis
%$\{(1, 0, 0),(a, b, 0),$ $(0, 0, f)\},$ where $f=\sqrt{c^2+d^2+e^2}.$ The geometric intuition comes from the fact that the faces of Voronoi cells of $\Lambda_p$ regarding its basis third vector lie in the same plane of the faces of its Babai cell and therefore, no further region is cut in this direction. On the other hand, if $c$ and $d$ are nonzero, the faces of the Babai cell of $\Lambda$ will provide another cut into its Voronoi cell. Hence, there is more volume in the intersection between Voronoi and Babai cells in the lattice $\Lambda$ than in $\Lambda_p.$
%Since the error probability of the prism lattice is the same of the associated two-dimensional lattice, we must have $P_e(\Lambda) \geq \frac{-a-a^2}{4b^2}.$

%Our interest is to discuss, as  was done in the two-dimensional case, the connections between error probability and packing density, particularly  the smallest error probability $P_e^\ast(\Delta)$, for lattices with a given packing density no smaller than $\Delta.$
Let $P_e$ and $\Delta_3$ be the error probability and packing density for a lattice $\Lambda$. Consider the curve $P_e^*(\Delta)$, the lower boundary of the set of points   $(\Delta_3,P_e)$  obtained by minimizing $P_e$ subject to the constraint $\Delta_3 \geq \Delta$. Our interest is in finding a parametric form for the  three-dimensional lattices that achieve points on this boundary. Observe that $P_e^*(\Delta)=0$, for $\Delta\leq \pi/6$, where $\pi/6$  is the packing density for the cubic lattice in three dimensions. In fact lattices with densities strictly smaller than $\pi/6$ and error probability equal to zero can be obtained by rectangular (i.e. cuboidal) lattices. However, since $P_e=0$ is already achieved at the packing density $\pi/6$, we need only consider  $\Delta$  in the range $\left[ \pi/6, \pi/(3\sqrt{2})\right]$, where $\pi/(3\sqrt{2})$ is the packing density of the FCC lattice, the lattice with the highest packing density in three dimensions.
It turns out that  a parametric form can be given, which closely approximates $P_e^*(\Delta)$, and coincides with it over a range of packing densities. This parametric form is obtained by placing some constraints on the parameters in the family of well-rounded lattices (defined in the sequel). 

\textit{Strongly well-rounded lattices},  are defined as lattices having a basis consisting of  vectors of minimum norm,  which  in our context is equal to $1$. %Note that for a lattice  with given angles between its basis vectors, the smallest volume and consequently highest packing density is achieved when all vectors have the same norm. %(in our case, equal to $1$). 
Well-rounded lattices have been studied generally \cite{DamirF:2019},~\cite{McMullen:2005}, and also for
applications such as coding for wiretap Gaussian and fading channels~\cite{Damiretal:2020, Gnilkeetal: 2016}. 

%The well-rounded lattices in our case have an obtuse superbase $\{(1,0,0),(\cos \alpha, \sin \alpha, 0),$ $(\sin \beta \cos \gamma, -\sin \beta \sin \gamma, \cos \beta)\},$ where $-\frac{1}{2} \leq \cos \alpha \leq 0,$ $-\frac{1}{2} \leq \sin \beta \cos \gamma \leq 0,$ and $-\frac{1}{2} \leq \sin \beta \cos(\alpha - \gamma) \leq 0.$ We will describe next families of lattices with densities varying from $\frac{\pi}{6}$ (of the cubic lattice $\mathbb{Z}^3$) to $\frac{\pi}{6} \sqrt{2}$ (of the FCC lattice).
%
%

The bases for the family of well-rounded lattices can be written as $\{(1,0,0),$ $(-\cos \alpha, \sin \alpha, 0),$ $(-\sin \beta \cos \gamma,-\sin \beta \sin \gamma,\cos \beta)\}$, with $-1/2 \leq -\cos \alpha \leq 0$, $-1/2\leq -\sin \beta \cos \gamma \leq 0$ and $-1/2 \leq \sin \beta \cos (\alpha + \gamma) \leq 0$. These bases are in Minkowski reduced form, and satisfy the superbase constraint. 
It turns out that $\Lambda(\beta)$, the  well-rounded lattice parameterized by $\beta$ with $\alpha=\pi/2$  and%$(1,0,0)$, $(0,1,0)$ and $(-\sin \beta \cos \gamma,-\sin \beta \sin \gamma,\cos \beta)$ where 
\begin{equation}
\sin \gamma=\left\{\begin{array}{cc}0, & 0 \leq \beta < \pi/6,\\ \frac{1}{2\sin \beta}, & \pi/6 \leq \beta \leq \pi/4, \end{array} \right.
\label{eqn:parametric3}
\end{equation}
leads to a curve  which closely approximates  $P_e^*(\Delta)$.
%The associated Gram matrix is
%\begin{eqnarray}
%A(\beta)=\left\{\begin{array}{lc}
%\begin{pmatrix} 1 & 0 & -\sin \beta \\ 0 & 1 & 0 \\ -\sin \beta & 0 & 1 \end{pmatrix}, &  0\leq \beta < \pi/6 \nonumber \\
%\begin{pmatrix} 1 & 0 & -1/2 \\ 0 & 1 & -\sqrt{\sin^2 \beta -1/4} \\ -1/2 & -\sqrt{\sin^ 2\beta -1/4} & 1 \end{pmatrix}, &  \pi/6 \leq \beta \leq \pi/4,
%\end{array} \right.
%\end{eqnarray}
%Note that for $0<\beta<\pi/6$, the Voronoi cell of the lattice $\Lambda(\beta)$ is a hexagonal prism.

Error probability -- packing density curves, obtained using the above parameterization, as well as a grid search, are plotted in the right hand panel in Fig.~\ref{fig-hp}. We have the following observations.
\begin{enumerate}
\item For $0\leq \beta \leq \pi/6$, $\Lambda(\beta)$ has basis $\{(1,0,0),(0,1,0),(-\sin \beta, 0, \cos \beta)\}.$ The  packing density $\Delta(\beta)=\pi/(6\cos \beta)$, varies between $\pi/6$ (cubic lattice)  and $\pi/(3\sqrt{3})$ (hexagonal lattice). The error probability is the same as for the two dimensional case and is given by $P_e=(1-(1+2\sin\beta)^2)/(16\cos^2 \beta)$, which is an increasing function of $\beta$ and lies in the range $[0,1/12]$. The Voronoi cell is a cube for $\beta=0$, a regular hexagonal prism for $\beta=\pi/6$ and an irregular hexagonal prism for $0<\beta < \pi/6$. From Fig.~\ref{fig-hp} it is evident that the parameterization is optimal for this range of $\beta$ values. It is interesting that there is no truly 3 dimensional Voronoi cell that is  is able to do better in this range.

\item For $\pi/6 \leq \beta \leq \pi/4$, $\Lambda(\beta)$ has basis $\{(1,0,0),(0,1,0),(-\sqrt{\sin^2 \beta -1/4}, -1/2, \cos \beta)\}.$
The  packing density $\Delta(\beta)=\pi/(6\cos \beta)$, varies between $\pi/(3\sqrt{3})$  and $\pi/(3\sqrt{2})$ (FCC). The error probability is   an increasing function of $\beta$ and lies in the range $[1/12,0.1505]$. The Voronoi cell is a hexarhombic dodecahedron for $\pi/6<\beta < \pi/4$ and a rhombic dodecahedron  for $\beta=\pi/4$. The parameterization coincides with $P_e^*(\beta)$ for only part of this range of $\beta$ values, but is a close approximation to $P_e^*(\Delta)$ over this entire range.
\end{enumerate}

 We also present an	 interesting comparison to a value listed  in
Tab.~\ref{lp}. Specifically, the lattice with basis  $\{(1,0,0),(0,1,0),(-\sqrt{17/108},-1/2,\sqrt{16/27})\}$
%$\{(1,0,0), (-\frac{1}{2}, \frac{\sqrt{3}}{2}, 0), (0, -\frac{\sqrt{17}}{9}, \frac{8}{9})\},$  
has the same volume and consequently the same packing density as the BCC lattice (whose Voronoi region is a truncated octahedron), but has error probability   $0.1368$ which is smaller than $0.1458$ achieved by the BCC lattice.

 At least in dimension $n=3$,  we have numerical evidence that when the packing density is small enough to be obtained by a prism, a prism is optimal.  An natural question is whether this observation holds  for dimensions greater than $3$, i.e. do prisms achieve  points on $P_e^*(\Delta)$ in higher dimensions, when $\Delta$ is small enough. The resolution of this is left as future work, since it will require the development of alternative analytic methods.

\section{Error Probability Estimation for Higher Dimensions}
\label{sec5}
Direct error probability calculations become increasingly difficult as the lattice dimension grows---we have already seen an example of this in going from $n=2$ to $n=3$ dimensions. Further, no parameterizations of lattices in very large dimensions are known, which makes it difficult to examine the tradeoff between the packing density and the error probability $P_e$. Thus it is more fruitful to obtain bounds using tools from probability theory, when $n$ becomes large. We first study the error probability under uniform probability distributions in Sec.~\ref{secunif} and under Gaussian distributions in Sec.~\ref{secgauss}.

\subsection{Uniform Distributions}
\label{secunif}
We need a few definitions. Let $S(r)$ be the Euclidean ball (sphere) of radius $r$ in $\mathbb{R}^n$ centered at the origin. %A lattice  $\Lambda$ is said to be a packing lattice lattice for $S(r)$ if $(\mbox{int}(S(r))+\bm{y})\bigcap (\mbox{int}(S(r))+\bm{y}')$ is the empty set for any distinct points $\bm{y},\bm{y}' \in \Lambda$ and a covering lattice for $S(R)$ if $\bigcup_{\bm{y}\in \Lambda}(\bm{y}+S(R))=\mathbb{R}^n$. Here we have used the notation $A+\bm{y}$, for a set $A\subset \mathbb{R}^n$ and vector $\bm{y} \in \mathbb{R}^n$, to denote the translation of the set, i.e. $A+\bm{y}=\{\bm{a}+\bm{y},~\bm{a} \in A\}$. %Let $r$ and $R$ denote the packing radius and the covering radius for a lattice $\Lambda$.
%Up to an isometry, if 
~The Babai cell of a lattice with a given basis is a hyperrectangle with sides of length $a_i>0$, $i=1,2,\ldots,n$ and we say that the Babai cell has size ${\bf a}=(a_1,a_2,\ldots,a_n) = (|v_{11}|,|v_{22}|, \ldots, |v_{nn}|),$ where $V$ is the upper triangular generator matrix of $\Lambda.$

Note that in this section we primarily work with $P_c=1-P_e$.

\begin{theorem}(\textit{A Chebyshev Bound})
Suppose lattice $\Lambda\subset \mathbb{R}^n$  has covering radius $r_{\text{cov}}$, a Babai cell of size ${\bf a}=(a_1,a_2,\ldots,a_n)$, and satisfies 
\begin{equation}
\frac{1}{12}\sum_{i=1}^na_i^2 > r_{\text{cov}}^2.
\label{eqn:condcheb}
\end{equation}
Then, for the uniformly distributed case, 
\begin{equation}
P_c=\mbox{Prob}(\bm{X}\in {\mathcal V}(0)\bigcap {\mathcal B}(0)|\bm{X} \in \cV(0)) \leq \frac{1}{180n^2} \frac{\sum_{i=1}^n a_i^4}{\delta^2}
\end{equation}
where $\delta=\frac{1}{n}\left(\frac{1}{12}\sum_{i=1}^na_i^2 -r_{\text{cov}}^2 \right).$
\end{theorem}

\begin{proof}
Note that $\mbox{Var}{X_i}=a_i^2/12$ and the $X_i$ are mutually independent. Let $\mu =(1/12n)\sum_{i=1}^n a_i^2$.  It follows that the event
	\begin{eqnarray}
	\left\{X \in S(r_{\text{cov}})\right\}   =   \left\{\frac{1}{n}\sum_{i=1}^nX_i^2 \leq \frac{r_{\text{cov}}^2}{n}\right\} 
	\subset \left\{\left|\frac{1}{n}\sum_{i=1}^nX_i^2 - \mu \right| > \underbrace{\mu-\frac{r_{\text{cov}}^2}{n}}_{\delta}\right\} .
	\end{eqnarray} 
	 Since $\mbox{Var}(X_i^2) = E[X_i^4]-E[X_i^2]^2 = {a_i^4}/{180}$, it follows by an application of the Chebyshev inequality that
	\begin{equation}
	\mbox{Prob}(X\in {\mathcal V}(0)) \leq \mbox{Prob}\left\{X \in S(r_{\text{cov}})\right\}  \leq \frac{\mbox{Var}((1/n)\sum_{i=1}^nX_i^2)}{\delta^2} = \frac{\sum_{i=1}^n a_i^4}{180n^2 \delta^2}.
	\end{equation}
\end{proof}

As an application of the theorem, consider the Barnes-Wall lattice $\Lambda_{16}\subset \mathbb{R}^{16}$ whose generator matrix is given in Fig. 4.10~\cite{conwaysloane}. From the generator matrix which is in lower triangular form, the Babai cell has size $(4,2^{(10)},1^{(5)})$  and the covering radius is known to be $\sqrt{3}$~\cite{conwaysloane}. An application of the above theorem gives $\mbox{Prob}(X\in {\mathcal V}(0)) \leq 0.539$. For the Leech lattice $\Lambda_{24}$, the size of the Babai cell is $(8, 4^{(11)}, 2^{(11)},1)$ and the covering radius is $\sqrt{2}$ which gives $\mbox{Prob}(X\in {\mathcal V}(0)) \leq 0.0833$. We also note that the theorem cannot be used for the lattice $E_8$, using the generator matrix given in~\cite{conwaysloane}, since the condition \eqref{eqn:condcheb} is not satisfied.

Unfortunately, the method does not apply to the family of lattices $A_n$. %$A_n$  has a standard $(n+1) \times n$ generator matrix ~\cite[pp.108-109]{conwaysloane}. However, Kim and Peters  expressed such 
$A_n$ has generator matrix in square form given by $V_{A_n} = I_n + \frac{c_n}{n}J_n,$ where $I_n$ is the  $n \times n$ identity matrix, $c_n=-1 \pm \sqrt{n+1}$ and $J_n$ is  $n \times n$ the matrix of ones~\cite{kim2010}. From this fact, we can determine the size of the Babai cell, i.e., the numbers $a_1, \dots, a_n,$ which are the the diagonal elements of the upper triangular matrix $R$ obtained through QR decomposition. Hence,
\begin{equation*}
a_1=r_{11} = \sqrt{2}, ~ ~ a_2=r_{22}=\sqrt{\frac{3}{2}}, ~ ~ a_3=r_{33}=\sqrt{\frac{4}{3}}=\frac{2}{\sqrt{3}}.
\end{equation*}
	
	If we move forward with this process, we get that the $k-th$ side of the Babai cell is $a_k=\sqrt{\frac{k+1}{k}},$ for any $k=1, \dots, n.$ Observe that the condition from Eq.~\eqref{eqn:condcheb} is not satisfied for this lattice. Indeed, 
\begin{equation}\label{eq:harmonic}
\frac{1}{12} \sum_{i=1}^{n} a_i^2 = \frac{1}{12} \sum_{i=1}^{n} \left( 1+\frac{1}{i} \right),
\end{equation}
and $r_{\text{cov}} = \dfrac{1}{\sqrt{2}} \left( \dfrac{2 \cdot \lfloor \tfrac{n+1}{2} \rfloor \left( n+1-\lfloor \tfrac{n+1}{2} \rfloor \right)}{n+1} \right)^{1/2}$ \cite[p. 109]{conwaysloane}. %where here $[x]$ stands for the integer part of $x.$ 
By considering the approximation for partial finite sum of the harmonic series together with Eq.~\eqref{eq:harmonic}, it is valid that
\begin{equation}
 \frac{1}{12} \sum_{i=1}^{n} \left( 1+\frac{1}{i} \right) \approx \frac{1}{12}(n+\log(n)+1) < r_{\text{cov}}^2, ~ ~ ~ \text{for all} ~ n.
\end{equation}

\begin{theorem} ({\textit Exclusion Bound}) For a lattice $\Lambda$ with covering radius is $r_{\text{cov}}$, suppose that a Babai cell has size $\bm{a}=(a_1, a_2,  \dots, a_n)$ which satisfies $a_1 \geq a_2 \geq  \ldots \geq a_m > 2r_{\text{cov}} \geq a_{m+1} \geq \ldots a_n$. Then 
\begin{equation}
P_c=\mbox{Prob}(\bm{X}\in {\mathcal V}(0)\bigcap {\mathcal B}(0)|\bm{X} \in \cV(0)) \leq \frac{(2r_{\text{cov}})^m}{\prod_{i=1}^m a_i}.
\end{equation}
When $m=0$, the bound is unity.
\end{theorem}
\begin{proof}
Without loss of generality assume that $\det \Lambda = 1$. The idea is to cut off parts of the Babai rectangle which are outside the sphere $S(r_{\text{cov}})$, starting with cutting planes $\pm r_{\text{cov}} \bm{e}_1\in \mathbb{R}^n$, where $\bm{e}_1=(1,0,\ldots,0)$. After the $i$th pair of cuts $\pm r_{\text{cov}}\bm{e}_i$, we are left with a smaller rectangle of size $(r_{\text{cov}},\ldots,r_{\text{cov}},a_{i+1},\ldots,a_n)$ which intersects $S(r_{\text{cov}})$. We stop after the $m$th pair of cuts, for then every face of the remaining rectangle intersects the interior of $S(r_{\text{cov}})$. The volume of the remaining rectangle is the desired upper bound on the probability. Thus
\begin{eqnarray}
P_c \leq  (2r_{\text{cov}})^m a_{m+1}\ldots a_n = \frac{(2r_{\text{cov}})^m}{a_1a_2\ldots a_m},
\end{eqnarray}
where in the last step we have used $a_1a_2\ldots a_n=\det \Lambda =1$.
\end{proof}
For the Barnes-Wall $\Lambda_{16}$ and Leech $\Lambda_{24}$ lattices, the corresponding values of $P_c$ are $0.866$ and $0.0078$ respectively. Similar to the Chebyshev bound, the Exclusion bound gives only a trivial result for the lattice $E_8$.

In fact the two bounds can sometimes be combined.
\begin{corollary}(\textit{Exclusion and Chebyshev bounds})
Suppose $m$ is defined as in the Exclusion bound and that $\delta_1=\frac{1}{12}\sum_{i=m+1}^na_i^2-r_{\text{cov}}^2(1-m/3) >0$. Then
$$P_c  \leq  \frac{(2r_{\text{cov}})^m}{\prod_{i=1}^m a_i} \left(\frac{m(2r_{\text{cov}})^4+\sum_{i=m+1}^na_i^4}{180 n^2 \delta_1^2} \right).$$
\end{corollary}
\begin{proof}
Direct application of the Exclusion bound followed by the Chebyshev bound.
\end{proof}
For the Barnes-Wall lattice and the Leech lattice this gives $\mbox{Prob}(X\in {\mathcal V}(0)) \leq 0.4854$ and $\mbox{Prob}(X\in {\mathcal V}(0)) \leq 4.314\times 10^{-4}$, respectively.

\subsection{Gaussian Distribution}
\label{secgauss}
%\textcolor{red}{Here we should be able to make the following kind of assertion. If the noise variance of the Gaussian noise is smaller than $xx$, then $P_e\to 0$ as $n\to \infty$.}
%	Our results through this paper are completely dependent on the assumption that the received vector $\bm{x} \in \mathbb{R}^n$ is uniformly distributed over the Babai partition. 
	
%	Suppose that $\bm{x}=V\bm{u}+\bm{\eta},$ where $\bm{x} \in \mathbb{R}^n$ is the received vector, $V \in \mathbb{R}^{n \times n}$ is the lattice generator matrix with full column rank, $\bm{u} \in \mathbb{Z}^n,$ and $\bm{\eta} \in \mathbb{R}^n$ is a noise vector following the Gaussian distribution $\mathcal{N}(0, \sigma^2I),$ with $\sigma^2$ being known. Assume that $V$ is upper triangular and the entries $v_{i,i}$ on the main diagonal are strictly positive. Then the success probability is given by the following theorem.
	
We now analyze the Gaussian case, as described in Sec.~\ref{sec:errprob} for which $P_e$ is given by \eqref{eqn:gaussprob}. Analytic evaluation of this probability in closed form is difficult, even in low dimensional cases. 

Numerical analysis of $P_e$ for $n=2$  as a function of the packing density for various values of the noise variance $\sigma^2$ is presented in Fig.~\ref{fig:Gauss2} (this is the counterpart of Fig.~\ref{fig:packden2} for the Gaussian case). For a two dimensional lattice $\Lambda$ with basis $\{(1,0),(a,b)\},$ we have calculated the term $T$ in \eqref{eqn:gaussprob}, which we will refer to here as $P_e(\sigma^2, a, b)$.%the Gaussian distribution over the `error triangles' (Fig.~\ref{rvectors}) normalized by the Gaussian distribution over the Babai region $\cB(0),$ which resulted in an error probability function $P_e(\sigma^2, a, b).$ 
We could observe that $\frac{\partial P_e(\sigma^2, a, b)}{\partial a} <0$ for $-\frac{1}{2} \leq a \leq 0$ and $b \geq \frac{\sqrt{3}}{2},$ therefore for a fixed variance $\sigma^2$ and fixed $b$, $P_e(\sigma^2, a, b)$ is decreasing with $a.$ Thus, the same minimization for the parameter $a$ done in Remark~\ref{remark_ep_pd} applies here. It is straightforward to conclude that smaller variance provides smaller error probability.

\begin{figure}[h!]
\begin{center}
		\includegraphics[height=6cm]{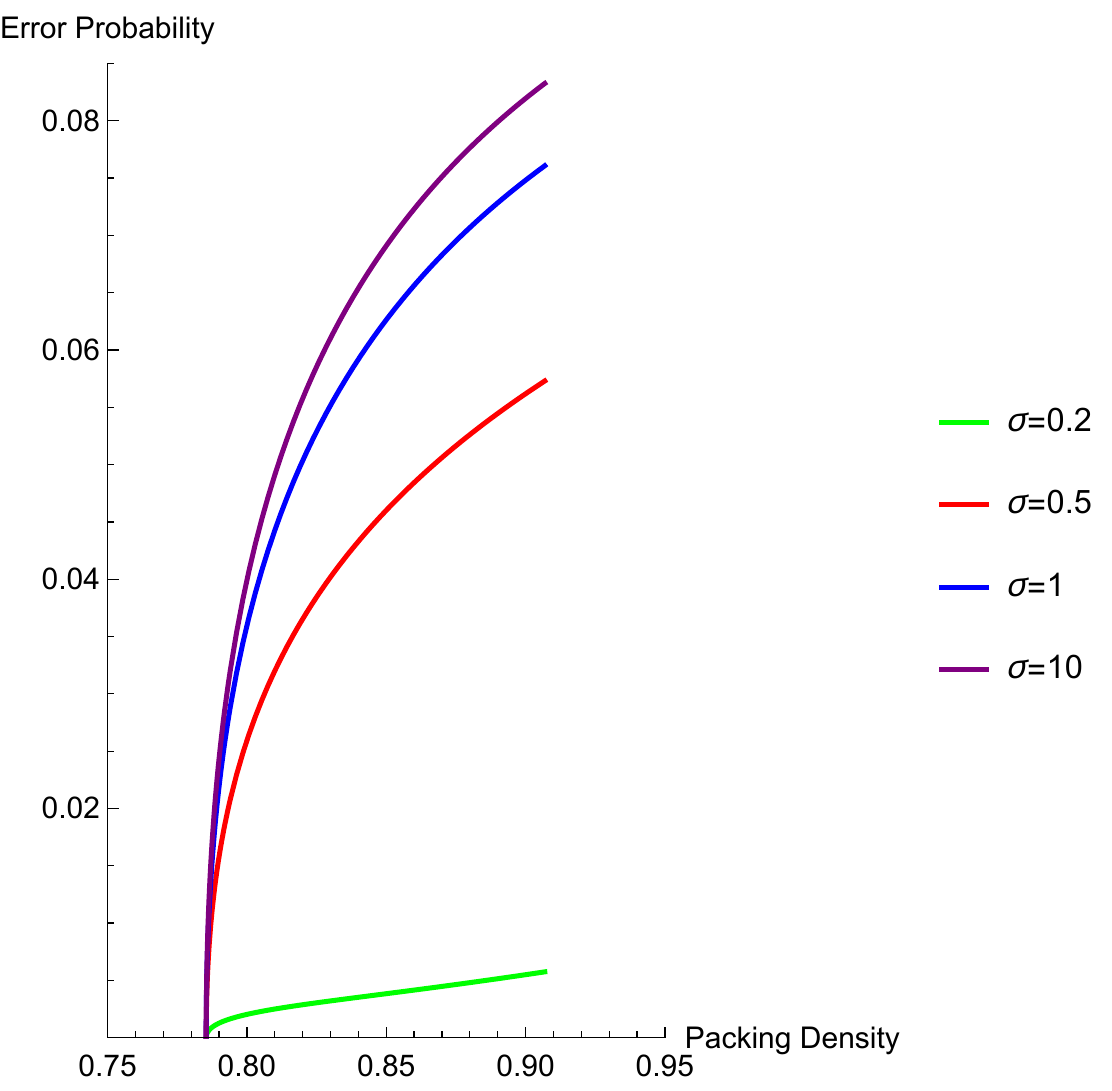} 
		\caption{Minimum error probability for given packing density assuming $\frac{\pi}{4} < \Delta_2 \leq \frac{\pi}{2\sqrt{3}} ~ (\text{or} ~ 3/4 \leq b^2 < 1),$ considering a Gaussian distribution. \label{fig:Gauss2} }
\end{center} 
\end{figure}

	Now, consider a lattice $\Lambda \subset \mathbb{R}^n,$ with $n$ large and its Voronoi cell $\cV.$ The largest radius among the radii of inscribed spheres in the Voronoi region $\cV$ is the packing radius $r_{\text{pack}}.$ Given that $\bm{Z} \sim \mathcal{N}(0, \sigma^2\bm{I})$, we are interested in calculating $\prob{\bm{Z} \in \cV(0) \bigcap \cB(0)}$.  Thus
\begin{equation}\label{prob1}
\prob{\bm{ Z} \in \cV(0)\bigcap \cB(0)} \geq \prob{\bm{Z} \in S({r_{\text{pack}}})\bigcap \cB(0)}, % \geq  \mbox{Prob}({\bf X} \in \mathcal{B}_{\sqrt{n \sigma^2}}),
\end{equation}
where recall that $S(r)$ denotes the $n-$dimensional ball centered at zero with radius $r$.

%	We are going to define $\mbox{Prob}(S,\sigma^2)$ as the probability that the zero-mean white-Gaussian noise with variance $\sigma^2$ lies inside $S,$ where $S$ is a general body in $\mathbb{R}^n.$ From Eq. \eqref{prob1}, we can rewrite the probability inequalities as
%\begin{equation}\label{prob2}
%\mbox{Prob}(\Lambda, \sigma^2) \geq  \mbox{Prob}(S({r_{\text{pack}}}),\sigma^2).% \geq  \mbox{Prob}(\mathcal{B}_{\sqrt{n \sigma^2}},\sigma^2)
%\end{equation}	
%
%	With all the notations defined, we can state the following result.
The following theorem provides a condition on $\sigma$ under which $P_c \to 1$ as $n \to \infty$.
\begin{theorem}\label{thm_probe}(\textit{Condition on $\sigma^2$ for success probability}) $\prob{\bm{ Z} \in \cV(0)\bigcap \cB(0)} \to 1$ as $n\to \infty$ for all $\sigma^2 < \frac{\text{vol}(\Lambda)^{2/n}}{4},$ if $r_{\text{pack}} \leq \frac{|a_i|}{2},$ for all $i = 1, \ldots, n,$ where $a_i$ are the sizes of the Babai cell.
\end{theorem}

\begin{proof} This follows from
\[
 \mbox{Prob}(||\bm{Z}|| < r_{\text{pack}}) = \chi^2_{CDF}\left(\frac{r_{\text{pack}}^2}{\sigma^2}; n \right),
\]
where $\chi^2_{CDF}(x;k)$ stands for the cumulative chi-squared distribution function with $k$ degrees of freedom. For $\chi^2_{CDF}(zn;n),$ if we take $z = \frac{r_{\text{pack}}^2}{\sigma^2 n},$ it is valid for large $n$, according to Corollary 7.2.2 \cite[p. 145]{zamir2014} that
{\small \begin{align*}
\chi^2_{CDF}(zn;n) \approx \begin{cases}
1, & z>1 \\
0, & z<1
\end{cases}.
\end{align*}}

	Since we want $  \prob{\bm{Z} \in S({r_{\text{pack}}})} \to 1 $ we must have that $z > 1 \Rightarrow$ and
\[
z= \frac{r_{\text{pack}}^2}{\sigma^2 n} > 1 \Rightarrow \sigma^2 < \frac{r_{\text{pack}}^2}{n} =\frac{d_{\min}^2(\Lambda)}{4n} = \frac{||\lambda||_{2}^2}{4n},
 \]
where $d_{\min}(\Lambda)$ is the minimum distance among all lattice points and $\lambda = \inf\{||x||_{2}: x \in \Lambda\setminus\{0\}\}.$ 
	
	Recall that from Minkowski theorem \cite{Cassels:1997}, one can upper bound the Euclidean norm of the shortest vector in a given lattice $\Lambda$ by $\sqrt{n}~\text{vol}(\Lambda)^{1/n}.$ Thus,
\[
\sigma^2 < \frac{||\lambda||_{2}^2}{4n} \leq \frac{(\sqrt{n}~\text{vol}(\Lambda)^{1/n})^{2}}{4n} = \frac{\text{vol}(\Lambda)^{2/n}}{4}.
\]
	
	Therefore, due to the fact that $r_{\text{pack}} \leq \frac{|a_i|}{2},$ for all $i,$ where $a_i$ denotes the sizes of the Babai cell, then for values of $\sigma^2 < \frac{\text{vol}(\Lambda)^{2/n}}{4},$ we can guarantee that $\prob{\bm{ Z} \in \cV(0)\bigcap \cB(0)} \geq \prob{\bm{Z} \in S({r_{\text{pack}}})} $ $\to 1$ as $n \to \infty$.
	
%	$\mbox{Prob}(\Lambda, \sigma^2) \geq  \mbox{Prob}(S({r_{\text{pack}}}),\sigma^2) \rightarrow 1$ and consequently $\mbox{Prob}(\Lambda, \sigma^2)\rightarrow 1,$ as we wanted to demonstrate.
\end{proof}

	Thm.~\ref{thm_probe} states that if the variance $\sigma^2$ satisfies the proposed condition, then estimating the Babai point is enough to guarantee the correct solution for the  nearest lattice point problem. In particular, examples where the hypothesis of Thm.~\ref{thm_probe} are satisfied includes the cubic lattice $\mathbb{Z}^n$ or rectangular lattices and we reach an analogous conclusion to the uniform case, i.e., that the error probability is vanishing for cubic (and rectangular) lattices.
	
%\begin{remark} Observe that the error probability of the Babai point with respect to the nearest lattice point is directly related to the packing density, when the received vector is  ${\bf X} \sim \mathcal{N}(0, \sigma^2).$ Indeed, as Eq.~\eqref{gaussian_vi} only depends on the diagonal elements of the generator matrix $V$ of the lattice $\Lambda,$ the off diagonal elements do not impact the error probability calculation. On the other hand, this is not true when ${\bf X} \sim \text{Unif}(\mathcal{B}(0)).$ To illustrate that, consider a two dimensional lattice with basis $\{(1,0),(a,1)\}$ where $-\frac{1}{2} \leq a \leq 0$ and notice that all lattices of this family have the same packing density $\Delta_2=\frac{\pi}{16}.$ Fig.~\ref{fig_remark} represents how the error probability varies in the uniform distribution and stays constant in the Gaussian ditribution.
%
%\begin{figure}[h!]
%\begin{center}
%		\includegraphics[height=5cm]{Plot_EP_a.png}  
%		\caption{Error probability and packing density assuming different distributions \label{fig_remark}}
%\end{center}
%\end{figure}
%\end{remark}
	
\section{Conclusions and Future Work}	
\label{secC}

We 	have considered the  problem of finding an approximate nearest point in a given lattice $\Lambda$ to $\bm{x} \in \mathbb{R}^n$  in a distributed network. We assumed that  each component of the vector $\bm{x}$ is available at a distinct sensor node and the  lattice point is to be obtained at a central node. Thus each sensor node sends a quantized version of its observation  to a central node. % Due to the limited capacity of the channel, the fusion center can only recover an approximation of the exact solution, which is the Babai point.
	
A protocol for transmitting this information to the central node is presented, its communication rate is determined, and is shown to be optimal when the components of $\bm{X}$ are mutually independent. We then consider the problem of evaluating the error probability, namely, the probability that the approximate nearest lattice point does not coincide with the nearest lattice point. Closed form expressions for the error probability  are derived in two dimensions. For the three dimensional case, using an obtuse superbase, we have estimated computationally for random lattices the worst error probability. For dimensions greater than $3$, we have obtained  bounds for the error probability.  Our results show that the error probability becomes larger as the packing density of the lattice becomes larger.  When the vector $\bm{x}$ is uniformly distributed over a certain region, it will be necessary to send extra bits to compute the nearest lattice point. However, when $\bm{x}$ is obtained by the addition of Gaussian noise of sufficiently small variance to a lattice point, no further communication will be necessary.

%\section*{Acknowledgments} We would like to thank you for \textbf{following
%the instructions above} very closely in advance. It will definitely
%save us lot of time and expedite the process of your paper's
%publication.

% You may incorporate your references as follows in your main tex file.
% Using BibTex is not recommended but can be handled.

\section{Acknowledgment}
CNPq (140797/2017-3, 312926/2013-8) and FAPESP (2013/25977-7) supported MFB and SIRC. VV was supported by  CUNY-RF and CNPq (PVE 400441/2014-4). We thank the reviewers of a previous draft for their constructive comments which helped improve the paper.%MFB would like to thank Nelson G. Brasil for meaninful discussions and contributions regarding to the computational implementation.

%	In the present work we explored a known (and hard) lattice problem under some new constraints, i.e., in a distributed network. We concluded that for two dimensions we can measure the error in the solution of the closest lattice point problem and improve this result by sending extra bits from each node. All the contributions we made depended on how good we can make a basis for a lattice and these methods only works for lower dimensions. Therefore, one question we can raise is: until what dimensions can we derive similar results, based on good basis for the respective lattices?
%	
%	Another direction is to explore the same problem in lattices where we know clearly their structure, for example,the families $A_{n} $ and $D_{n}.$ What can we generalize to these category of lattices by using the information we derived in this paper?
%	

\end{document}